\documentclass[12pt]{article}
\usepackage{amsmath}
\usepackage{amsfonts}
\usepackage{amssymb}
\usepackage{amsthm}
\usepackage{graphicx}
\usepackage{fullpage}
\usepackage{setspace}
\usepackage{verbatim}
\usepackage{natbib}
\usepackage{appendix}
\usepackage{multirow}

\newtheorem{theorem}{Theorem}[section]
\newtheorem{lemma}{Lemma}[section]

\newtheorem{assumption}{Assumption}[section]

\newcommand{\cX}{{\cal X}}
\newcommand{\mX}{\mathbb{X}}

\title{Multiscale Adaptive Inference on Conditional Moment Inequalities
\footnotetext{
We thank Tze Leung Lai, Joe Romano, Han Hong, Don Andrews, Xiaohong Chen
and seminar participants at Yale, UT Austin, Northwestern and
the 2012 Greater New York Metropolitan Econometrics Colloquium for helpful comments and
Bonny Wang for research assistance.
Armstrong acknowledges support from a
fellowship from the endowment in memory of
B.F. Haley and E.S. Shaw through SIEPR.}}

\author{Timothy B. Armstrong  \\
Yale University
\thanks{email: timothy.armstrong@yale.edu}
\and
Hock Peng Chan  \\
National University of Singapore
\thanks{email: stachp@nus.edu.sg}
}

\begin{document}

\maketitle

\begin{abstract}
This paper considers inference for conditional moment inequality models using a multiscale statistic.  
We derive the asymptotic distribution of this test statistic and use the result to propose feasible critical values that have a simple analytic formula,
and to prove the asymptotic validity of a modified bootstrap procedure.
The asymptotic distribution is extreme value, and the proof uses new techniques to overcome several technical obstacles. 
The test
detects local alternatives that approach the identified set at the best 
rate
among available tests
in a broad class of models, and
is adaptive to the smoothness properties of the data generating process.  Our results also have implications for the use of moment selection procedures in this setting.
We provide
a monte carlo study and
an empirical illustration to inference in a regression model with endogenously censored and missing data.
\end{abstract}

\section{Introduction}

This paper considers inference in conditional moment inequality models based on a multiscale test statistic with certain adaptive power properties.
Formally, the model is defined by a vector of inequality restrictions of the form $E(m(W_i,\theta)|X_i)\ge 0$ almost surely, where $m$ is a known parametric function and inequality is taken elementwise.  The set $\Theta_0$ of parameter values that satisfy this set of restrictions is called the identified set, and the goal is to form a test that has good power properties at alternative values of $\theta$ near the boundary of the identified set.  By testing the null $\theta\in\Theta_0$ for each $\theta$, and inverting these tests, one obtains a confidence region that, for each point in the identified set, contains this point with a prespecified probability \citep[see][for a discussion of this and other notions of inference in this setting]{imbens_confidence_2004}.   This class of models includes numerous models used in empirical economics, including selection models, regression models with endogenously missing or censored data, and certain models of firm and consumer behavior (see below for references from the literature).

We derive the asymptotic distribution of our test statistic and show how it can be used to obtain feasible critical values.
These critical values have
the advantage of having a simple analytic formula that can be computed without using simulation.
This is particularly useful in applied settings where computational issues can severely limit the applicability of tests that require resampling or simulation to compute critical values.
We also prove the asymptotic validity of a modified bootstrap procedure, which we consider in an appendix.
While we focus on least favorable critical values, both methods can be used with first stage moment selection procedures.
We provide power results that show that our test detects alternative parameter values that approach the boundary of the identified set at the fastest rate among procedures currently available in the literature.
While the power results in this paper are stated for a single underlying distribution  and sequence of parameter values satisfying certain conditions, these power comparisons can also be shown to hold in a minimax sense (see Appendix \ref{power_comp_sec} for a detailed discussion and references to the literature).
The test is adaptive in the sense that it achieves these rates for data generating processes with a range of smoothness properties without prior knowledge of these smoothness properties.
The test achieves these rates adaptively even without the use of first stage moment selection procedures, and our results show that moment selection procedures have little or no first order effect on power in many settings.
While moment selection procedures will have some effect in finite samples, the results suggest that our test is less sensitive to moment selection than many of the procedures available in the literature.
This is a particularly positive result for researchers who prefer not to use pre-tests because of computational issues, or because of the introduction of arbitrary user driven parameters.

The test statistic we consider presents several technical obstacles in deriving the asymptotic distribution.  Because of the variance weighting, which is needed for our test to have good power properties, the test statistic takes a supremum over a sequence of random processes for which functional central limit theorems do not hold.  While similar technical issues have been solved in other settings using approximations by sequences of gaussian processes \citep[see, for example][]{bickel_global_1973,chernozhukov_intersection_2009}, the multiscale nature of our test statistic (as opposed to test statistics based on kernels with a fixed sequence of bandwidths), makes the rate of approximation too poor for our purposes
(see Appendix \ref{gauss_approx_sec}).
In addition, the test statistic we consider takes the supremum over a process that is nonstationary in ways that the previous literature has not dealt with, so even deriving the asymptotic distribution of the supremum of the approximating gaussian process would require new techniques.

To overcome this, we use methods for tail approximations to nonstationary, nongaussian processes, applying them directly to the process in the sample.  We use methods from \citet{chan_maxima_2006} to derive tail approximations directly using a combination of moderate deviations results and tail equicontinuity conditions, thereby circumventing the need for strong approximations.  We verify these conditions for our test statistic directly, and use these results in the derivation of the extreme value distribution.  While verifying these conditions can be challenging, we anticipate that the techniques introduced here will be useful in other problems in econometrics.

\subsection{Related Literature}\label{lit_subsec}

This paper is related to the literature on partial identification and, in particular, the literature on conditional moment inequalities.  The tests proposed in this paper are most closely related to those studied by \citet{armstrong_weighted_2011},
\citet{armstrong_weighted_2014}
and \citet{chetverikov_adaptive_2012}
(the results in the present paper were developed independently and around the same time as the latter paper).
\citet{armstrong_weighted_2011,armstrong_weighted_2014} considers estimation of the identified set using conservative confidence regions.  While those results could be used for the problem considered here, the methods of proof used in that paper lead to extremely conservative critical values that are too large to be useful in most practical settings.
\citet{chetverikov_adaptive_2012}
uses a different form of a statistic similar to ours (the supremum is taken only over a finite set of bandwidths and points that cannot grow too quickly) and different methods of proof that avoid deriving an asymptotic distribution or even showing that one exists.
From a practical perspective, our method delivers an analytic formula that can be used to compute a critical value that does not require simulation, and also proves the asymptotic validity of modified bootstrap procedures, while the approach taken in \citet{chetverikov_adaptive_2012} only allows for the latter result.
The analytic formula for the critical value also allows for more precise power results, both for the bootstrap and non-bootstrap version of the procedure.
On the other hand, the method in \citet{chetverikov_adaptive_2012} allows for better conditions for moment selection procedures.  (While we do not consider moment selection explicitly, our methods could be extended to this case.  However, the rate at which the set of selected moments can shrink is inherently constrained by our methods.  See Section \ref{mom_sel_sec} for more on moment selection procedures).  The methods used in that paper also give higher order coverage results for the bootstrap procedure (while extensions to our method have been shown to give higher order improvements in other contexts, we do not pursue this in this paper; see Appendix \ref{other_cval_sec}).

Papers proposing other approaches to inference on conditional moment inequalities include
\citet{andrews_inference_2013},
\citet{kim_kyoo_il_set_2008}, \citet{khan_inference_2009},
\citet{chernozhukov_intersection_2013},
\citet{lee_testing_2013},
\citet{ponomareva_inference_2010},
\citet{menzel_estimation_2008} and \citet{armstrong_asymptotically_2011}.
While these approaches are useful in many settings (for example, settings where point identification is likely, or where the researcher has prior knowledge of certain smoothness properties of the data generating process), they do not achieve optimal power adaptively in the generic set identified case considered here.
See Appendix \ref{power_comp_sec} for details and a formal statement.

This paper is also related to the broader literature on partial identification, including the problem of inference on finitely many unconditional moment inequalities.
Articles that consider this problem %
include
\citet{andrews_confidence_2004}, %
\citet{andrews_inference_2008}, %
\citet{andrews_validity_2009}, %
\citet{andrews_inference_2010}, %
\citet{chernozhukov_estimation_2007}, %
\citet{romano_inference_2010}, %
\citet{romano_inference_2008}, %
\citet{hansen_test_2005},
\citet{bugni_bootstrap_2010},
\citet{beresteanu_asymptotic_2008},
\citet{moon_bayesian_2009},
\citet{imbens_confidence_2004}
and \citet{stoye_more_2009}.
In addition, there have been a number of applications of partial identification, including the conditional moment inequality models considered here, going back at least to \citet{manski_nonparametric_1990}.  There are too many references to name all of them here, but papers include
\citet{pakes_moment_2006},
\citet{manski_inference_2002}, and
\citet{ciliberto_market_2009}.

From a technical standpoint, this paper is related to other papers deriving extreme value results for supremum statistics.  The literature goes back at least to \citet{bickel_global_1973}, and includes recent papers such as \citet{chernozhukov_intersection_2009}.  The arguments used in the proof in this paper are substantially different, as they do not use intermediate approximations by gaussian processes.  As discussed in more detail in Section \ref{setup_asym_dist_sec}, the multiscale nature of the test statistic considered here makes the rates in these approximations too poor for our purposes.  Our result also differs in that the test statistic we consider takes a supremum over a process that is nonstationary in ways not considered in the previous literature.  While extreme value results have been derived for nonstationary processes \citep[see, for example,][]{lee_testing_2009}, these results use other aspects of the structure of these problems that do not apply in our case.

The test statistic considered in this paper is related to scan statistics considered in the statistics literature.
This paper is also related to the literature on adaptive inference.  In particular, \citet{dumbgen_multiscale_2001} apply a similar approach to ours in a one dimensional gaussian setting.
This paper contributes to these literatures by deriving extreme value approximations in a setting with a multidimensional, nongaussian, nonstationary process, which requires new techniques for the same reasons described above.
\citet{spokoiny_adaptive_1996} and
\citet{horowitz_adaptive_2001} propose different tests for a related goodness of fit testing problem.
Those authors consider adaptivity with respect to a different class of alternatives than the one in this paper, leading to a different approach.  In particular, \citet{horowitz_adaptive_2001} consider minimax rates with respect to $L_2$ distance in a two-sided testing problem.  Our test is taylored toward the goal of inverting the test to form a confidence region for the parameter $\theta$, and has good power properties when one considers Euclidean distance of alternative parameter values $\theta$ to the identified set $\Theta_0$ (see Appendix \ref{power_comp_sec} for further discussion).

\subsection{Notation and Plan for Paper}

We use the following notation throughout the rest of the paper.  For observations
$\{Z_i\}_{i=1}^n$, the sample mean of a function $g$ is given by
$E_ng(Z_i)\equiv \frac{1}{n}\sum_{i=1}^n g(Z_i)$.  Inequalities are defined for vectors as holding elementwise and, for a vector $x$ and a scalar $b$, we write $x\ge b$ iff. all components of $x$ are greater than equal to $b$.  For vectors $a$ and $b$, $a\wedge b$ is the elementwise minimum, and $a\vee b$ is the elementwise maximum.
We use $a_n\sim b_n$ to denote the statement that $a_n/b_n\to 1$.

The rest of the paper is organized as follows.  Section \ref{setup_asym_dist_sec} describes the setup and gives the main asymptotic distribution result.  Section \ref{inference_sec} derives critical values for the test based on this result.  Section  \ref{power_sec} provides results on the power of the test.
Section \ref{monte_carlo_sec} reports the results of a monte carlo study.
Section \ref{application_sec} reports the results of an illustrative empirical application.
Section \ref{conclusion_sec} concludes.
Appendices to the main text contain proofs of the results in the main text, as well as some additional results mentioned in the main text, including versions of some of the results from the body of the paper that incorporate uniformity in the underlying distribution and a comparison of the power properties of the test with other procedures in the literature.

\section{Setup and Asymptotic Distribution}\label{setup_asym_dist_sec}

We observe iid data $\{X_i,W_i\}_{i=1}^n$ where
$X_i\in\mathbb{R}^{d_X}$ and $W_i\in\mathbb{R}^{d_W}$.  We wish to test the null
hypothesis
\begin{align}\label{null_eq}
E(m(W_i,\theta)|X_i)\ge 0 \,\,\,\,\,\,\,\, \text{a.s.}
\end{align}
where
$m:\mathbb{R}^{d_W}\times \Theta\to\mathbb{R}^{d_Y}$ is a known measurable
function and $\theta\in\Theta\subseteq\mathbb{R}^{d_\theta}$ is a fixed
parameter value.  We use the notation $\bar m(\theta,x)$ to denote a
version of $E(m(W_i,\theta)|X_i=x)$.  Typically, the null (\ref{null_eq}) is tested for each value of $\theta$ in order to obtain a confidence region for parameters that are consistent with the model.  The model may not be point identified, in the sense that there may be more than one value of $\theta$ consistent with (\ref{null_eq}), and the tests in this paper are specifically geared towards this case.  In general, we denote by $\Theta_0$ the identified set of parameter values that are consisent with the restrictions in (\ref{null_eq}):
\begin{align*}
\Theta_0\equiv \left\{\theta\in\Theta|
E(m(W_i,\theta)|X_i)\ge 0  \text{ a.s.}
  \right\}.
\end{align*}
While the above setup considers only a single probability distribution, this is only for notational convenience.  We show in Appendix \ref{uniformity_sec} that our test controls the asymptotic size uniformly over appropriate classes of underlying distributions.

We note that, while the above setup is written in terms of a parametric model $m(W_i,\theta)$, our methods apply more generally to test the inequality $E(Y_i|X_i)\ge 0$ a.s., where $Y_i$ is any random variable satisfying certain regularity conditions below.  The reason we impose this additional structure is that our tests are designed to have good power properties for values of $\theta$ that violate the null, but are near the identified set $\Theta_0$ of parameters that satisfy the null.  Since our goal is to distinguish parameter values in $\Theta_0$ from nearby parameter values outside of $\Theta_0$, we state our power results in terms of sequences of parameter values and the rate at which they approach the boundary of $\Theta_0$ (see Section \ref{power_sec}).  By deriving our results in terms of alternative parameter values rather than abstract notions of distances of data generating processes, we obtain power results that are immediately applicable to assessing the statistical accuracy of confidence regions based on our tests in economic models (see Appendix \ref{power_comp_sec} for further discussion).

Consider the test statistic $T_n=(T_{n,1},\ldots,T_{n,d_Y})$ where
\begin{align*}
T_{n,j}=T_{n,j}(\theta)
  \equiv \left| \inf_{I(s,t)\subseteq \hat{\mathcal{X}}, t\ge t_n}
    \frac{E_nm_j(W_i,\theta)I(s<X_i<s+t)}{\hat\sigma_{n,j}(s,t,\theta)}
     \right|_{-},
\end{align*}
$t_n$ is a sequence of scalars going to zero (the condition $t\ge t_n$ is interpreted as stating that all components of $t$ are greater than or equal to $t_n$), $\hat{\mathcal{X}}$ is the convex hull of $\{X_i\}_{i=1}^n$, $I(s,t)=[s_1,s_1+t_1)\times\cdots\times [s_{d_X},s_{d_X}+t_{d_X})$ and
\begin{align*}
\hat\sigma_{n,j}^2(s,t,\theta)
  \equiv E_n m_j(W_i,\theta)^2I(s<X_i<s+t)-[E_n m_j(W_i,\theta)I(s<X_i<s+t)]^2.
\end{align*} 
We can form a test by rejecting for large values of $S_n=S(T_n)$, where $S:\mathbb{R}^{d_Y}\to\mathbb{R}$ is some function that is nondecreasing in each argument.  For concreteness, we take $S$ to be function that takes the maximum of the components of $T_n$:
\begin{align*}
S_n=S_n(\theta)=\max_{1\le j\le d_Y} T_{n,j}(\theta).
\end{align*}

It is worth commenting on the properties of this test statistic that differ from other statistics for this problem, and how they lead to optimal power properties for set identified models.  We discuss this briefly here, and refer the reader to Appendix \ref{power_comp_sec} and \citet{armstrong_choice_2014} for details.  In testing $E(m(W_i,\theta)|X_i)\ge 0$ a.s., one can use essentially any test statistic that estimates $E(m(W_i,\theta)|X_i)$ and takes some function of this that is large in magnitude when this estimate is negative for some value of $x$.  Most conditional mean estimates can be thought of as using an instrumental variables approach, where the inequality $E(m(W_i,\theta)|X_i)\ge 0$ a.s. is transformed into a set of inequalities $Em(W_i,\theta)g(X_i)\ge 0$ where $g$ ranges over a set $\mathcal{G}_n$ that is infinite or increases with the sample size (e.g., a kernel estimator does this with the functions $g$ given by $h((X_i-x)/h_n)$ where $h_n$ goes to zero at some rate and $x$ ranges over the support of $X_i$) and the inequality may only hold approximately if $g$ is not positive everywhere (e.g. if higher order kernels or sieves are used).  Once a class $\mathcal{G}_n$ is decided on, one faces the decision of how to transform estimates of $Em(W_i,\theta)g(X_i)$ into a statistic that is positive and large in magnitude whenever one of these estimates is negative and large in magnitude.  This includes deciding on how to weight each function $g$, and how to combine them.  For the latter problem, one can take some power of the negative part of the test statistic and add or integrate these over $g$ (a Cramer-von Mises or CvM style approach), or take the maximum or supremum of the negative part (a Kolmogorov-Smirnov or KS approach).  In addition, since the null space is composite, one faces a choice in how to pick the critical value, and, in particular, whether to choose a critical value based on the least favorable distribution in the null space where $E(m(W_i,\theta)|X_i=x)=0$ for all $x$, or whether to use a pre-testing procedure that determines where the equality may hold and uses smaller critical values based on the results of this procedure.

In sum, one faces the decision of (1) which instruments (or kernels or sieves, etc.) to use, (2) how to weight them, (3) how to combine them (integration or summing, or taking the supremum) and (4) how to choose the critical value.  For (1), our test statistic uses a class of product kernels with all possible bandwidths above a cutoff.  Using a class of functions with multiple scales, rather than a kernel function with a single bandwidth, allows the test to find the optimal bandwidth adaptively for a range of smoothness conditions.  For (2), the test statistic $S_n$ weights each function by its standard deviation.  This weighting is essential in allowing the test statistic to find the instrument function that balances bias and variance in an optimal way for detecting a given alternative, and the improvement in power in the set identified case can be thought of as an optimal weighting result for moment inequality models.

For (3) our test statistic uses a supremum (KS) criterion rather than a criterion based on sums or integrals (a CvM criterion).  To understand why a KS approach leads to more power than a CvM approach for the alternatives considered here, it is helpful to consider the relationship between the nonsimilarity of these tests on the boundary of the identified set and power at nearby alternatives.  If a test statistic behaves differently depending on where $\bar m(x,\theta)=0$, then using the most conservative critical value will lead to poor power in cases where nearby parameter values in the null space lead to the inequality binding on a small set.
While moment selection procedures can help alleviate this, they can be computationally costly, and the versions of these procedures proposed in the literature often contain tuning parameters that prevent the critical value from being too small under alternatives of the form considered in this paper (e.g. \citealp{andrews_inference_2013} introduce a tuning parameter that prevents their critical value from shrinking to zero at a faster rate than $\sqrt{n}$ which, as shown by \citealp{armstrong_asymptotically_2011}, leads to a decrease in the rate at which local alternatives can approach the identified set and still be detected).
KS statistics are less sensitive to which moments bind since the supremum of $k$ sample means increases at a $\sqrt{\log k}$ rate, while the sum of the positive part increases at a polynomial rate in $k$.  Thus, by using a KS criterion, our test statistic achieves good power without requiring moment selection procedures, and the power of the test is less sensitive to these procedures, so that the decision (4) has less impact on the power of the test.

We impose the following conditions.

\begin{assumption}\label{asym_dist_assump}
~
\begin{itemize}
\item[a.)]  The distribution of $m(W_i,\theta)$ conditional on $X_i$ satisfies the following conditions.
\begin{itemize}
  \item[i.)] There exists a $\lambda>0$ and a constant $M_\lambda$ such that %
    \begin{align*}
      E(\exp(\lambda |m_j(W_i,\theta)|)|X_i)< M_\lambda \text{ a.s. all $1\le j\le d_Y$}.
    \end{align*}
  
  \item[ii.)] $var(m_j(W_i,\theta)|X_i=x)$ is positive and continuous in $x$ for all $j$.

  \item[iii.)] $corr(m_j(W_i,\theta),m_k(W_i,\theta)|X_i=x)$ is bounded away from $1$ for all $j\ne k$.
  \end{itemize}
\item[b.)]
The support $\mathcal{X}$ of $X_i$ is a compact, convex
Jordan measurable set with strictly positive measure, and $X_i$ has a density $f$ that is bounded away from zero on $\mathcal{X}$.
\item[c.)]
$t_n\to 0$ and $nt_n^{d_X}/|\log t_n|^4\to\infty$.
\end{itemize}
\end{assumption}

Part (a) imposes regularity conditions on the moments of $m(W_i,\theta)$.
It is worth noting that, while we impose some mild smoothness assumptions on the conditional variance, we place no assumptions on the smoothness of the conditional mean.  Thus, while the power of our test depends on the smoothness properties of the conditional mean, our test is robust to very nonsmooth data generating processess.  The convexity assumption in part (b) is imposed to simplify certain parts of the proof, and could be relaxed.
Note that, while part (b) rules out cases where $X_i=(X_{i,1}',X_{i,2}')'$, where $X_{i,1}$ is continuously distributed and $X_{i,2}$ is discretely distributed on some set $\{x_1,\ldots, x_k\}$, this can be accomodated by redefining $X_i$ to be $X_{i,1}$, redefining $W_i$ to be $(W_i',X_{i,2})$, and redefining $m$ to be the $\mathbb{R}^{d_Y\cdot k}$-valued function with $d_Y\cdot (\ell -1)+j$th component given by
$m_j(W_i,\theta)I(X_{i,2}=x_\ell)$.
The condition on $t_n$ in part (c) is, up to the $|\log t_n|$ term, the best possible rate (see Section \ref{gauss_approx_sec}).

The following theorem gives the asymptotic distribution of this test statistic, and provides feasible critical values that can be calculated analytically.  For a version of this theorem that incorporates uniformity in the underlying distribution, we refer the reader to Appendix \ref{uniformity_sec}.

\begin{theorem}\label{asym_dist_thm_randx}
Suppose that the null hypothesis (\ref{null_eq}) and Assumption \ref{asym_dist_assump} hold for $\theta$.
Let $\hat c_n=vol(\hat{\mathcal{X}})/t_n^{d_X}$ and let
$a(\hat c_n)=\left(2n \log \hat c_n\right)^{1/2}$ and
$b(\hat c_n)=2\log \hat c_n+(2d_X-1/2)\log\log \hat c_n-\log (2\sqrt{\pi})$.
Then, for any vector $r\in\mathbb{R}^{d_Y}$,
\begin{align*}
\liminf_n P\left(a(\hat c_n) T_n-b(\hat c_n) \le r \right)
  \ge P(Z\le r) 
\end{align*}
where $Z$ is a $d_Y$ dimensional vector of independent standard type I extreme value random variables.  If, in addition
$\bar m_j(\theta,x)=0$ for all $x\in\mathcal{X}$ and $j=1,\ldots,d_Y$, then 
\begin{align*}
a(\hat c_n) T_n-b(\hat c_n) \stackrel{d}{\to} Z.
\end{align*}
\end{theorem}

\section{Inference}\label{inference_sec}

An immediate consequence of Theorem \ref{asym_dist_thm_randx} is a method for choosing feasible critical values for the test statistic $S_n(\theta)$ that can be computed analytically.  By Theorem \ref{asym_dist_thm_randx}, $a(\hat c_n)S_n-b(\hat c_n)$ is asymptotically bounded by a random variable that is the maximum of $d_Y$ standard type I extreme value random variables.  By the properties of extreme value random variables, this distribution is itself type I extreme value, with cdf $\exp(-d_Y\exp(-r))$.  Some calculation leads to the rejection rule
\begin{align}\label{cval_eq}
\text{reject if } S_n(\theta)>\hat q_{1-\alpha} \,\,\,\text{ where }\,
\hat q_{1-\alpha}
\equiv\frac{\log(d_Y)-\log(-\log(1-\alpha))+b(\hat c_n)}{a(\hat c_n)}.
\end{align}
It follows from Theorem \ref{asym_dist_thm_randx} that this test is asymptotically level $\alpha$.
We record this result in the following theorem.

\begin{theorem}\label{test_size_thm}
Suppose that the null hypothesis (\ref{null_eq}) holds for $\theta$ and that Assumption \ref{asym_dist_assump} holds.  Let
$\hat q_{1-\alpha}$ be as defined in (\ref{cval_eq}).
Then
\begin{align*}
\limsup_n P\left(S_n(\theta)>\hat q_{1-\alpha}\right)
  \le \alpha.
\end{align*}
If, in addition, $\bar m_j(\theta,x)=0$ for all $x\in\mathcal{X}$ and $j=1,\ldots,d_Y$, then
\begin{align*}
P\left(S_n(\theta)>\hat q_{1-\alpha}\right)
  \to \alpha.
\end{align*}
\end{theorem}

While the critical value given in (\ref{cval_eq}) gives a valid asymptotically level $\alpha$ test, this critical value is based on extreme value approximations that may perform poorly in finite samples in certain situations.
While our monte carlos suggest that the analytic critical values perform well in many cases encountered in practice, we
consider other methods, including a bootstrap or simulation based approach, in Appendix \ref{other_cval_sec}.
While our results in Appendix \ref{other_cval_sec} do not give a formal result showing an improvement in coverage accuracy, similar methods have been shown to lead to higher order improvements in the coverage accuracy in other settings (see Appendix \ref{other_cval_sec} for details and references to the literature).

\subsection{Moment Selection Procedures and the Choice of $t_n$}\label{mom_sel_sec}

The rejection probabilities of the tests defined above will converge to $\alpha$ when the conditional mean $\bar m_j(\theta,x)$ is equal to zero for all $x\in\mathcal{X}$ for all $j$.  %
If these inequalities only bind on a subset of $\mathcal{X}$, the rejection probability will be strictly less than $\alpha$, and it would seem that there would be the potential for large power improvements at nearby alternatives by using a smaller critical value that take this into account.  Perhaps surprisingly,
it turns out that there will be no first order power improvement from doing this in cases where the subset on which the conditional moments bind has positive probability
\citep[in certain cases where the binding subset has zero probability, our test loses a $\log n$ term in the rate at which local alternatives can approach the identified set and be detected, while certain other approaches lose a polynomial term; see][]{armstrong_asymptotically_2011,armstrong_weighted_2014}.
While this result should certainly not be taken to mean that the effect on power will be always be negligible in finite samples, the result suggests that our procedure will be less sensitive to moment selection than other procedures in the literature for which moment selection has a large effect on power asymptotically \citep[see, for example][]{armstrong_asymptotically_2011,andrews_inference_2013}.

To see why this holds, first, note that, it can be shown that,
if for some set $\tilde{\mathcal{X}}$,
$\bar m_j(\theta,x)>0$ for all $x\not\in \tilde{\mathcal{X}}$ and $j=1,\ldots,d_Y$,
the first display of
Theorem \ref{asym_dist_thm_randx} will hold with $\tilde{\mathcal{X}}$ replacing $\mathcal{X}$.  Thus, if we use prior knowledge of such a set $\tilde{\mathcal{X}}$ with strictly positive volume, or find such a set with a first stage test, we would obtain a critical value $\hat q_{1-\alpha}$ with $\mathcal{X}$ replaced by $\tilde{\mathcal{X}}$.  But note that, regardless of $\tilde{\mathcal{X}}$,
we have, letting $\tilde q_{1-\alpha}$ be the critical value formed with
$vol(\tilde{\mathcal{X}})$ in place of
$vol(\hat{\mathcal{X}})$
\begin{align*}
\frac{\tilde q_{1-\alpha}}{
(2\log t_n^{-d_X})^{1/2}/n^{1/2}}
\to 1.
\end{align*}
Thus, even with prior knowledge of the contact set, the contact set would have only a second order effect on the critical value.

The above calculations can also be used to understand the effect of the choice of the minimal window width $t_n$ on the power of the test.  Suppose that $t_n$ is chosen proportional to $n^{-\delta}$ for some $0<\delta<1$.  Then, by the same calculations, we will have
\begin{align*}
\frac{\hat q_{1-\alpha}}{(2 d_X \delta\log n)^{1/2}/n^{1/2}}
\stackrel{p}{\to} 1.
\end{align*}
As shown in Section \ref{power_sec},
larger values of $\delta$ are required to obtain optimal power properties for less smooth conditional means.  While choosing a larger value of $\delta$ does not affect the rate at which local alternatives can approach the null space and be detected (the test is adaptive with $t_n$ decreasing as quickly as allowed), it does have a non negligible effect on power through larger critical values.  If $t_n$ is chosen as $n^{-\delta_2}$ for some value $\delta_2$ instead of some other value $\delta_1$ where $\delta_1>\delta_2$, the critical value will increase by a factor of $(\delta_1/\delta_2)^{1/2}$.

It can also be shown that,
letting $\bar q_{1-\alpha}$ be the critical value for a test that only takes the infimum over all $s$ with $t$ fixed at $t_n$ (i.e. a test that uses the kernel approach considered in \citet{chernozhukov_intersection_2013} and \citet{ponomareva_inference_2010}), we have $\hat q_{1-\alpha}/\bar q_{1-\alpha}\to 1$.
Thus, considering larger bandwidths only has a second order effect on the critical value when using the multiscale approach in this paper.
By examining the critical values derived in \citet{chernozhukov_intersection_2013}, it can also be seen that the discussion above regarding the effect of ``moment selection'' on the critical value applies to these tests as well: the critical value converges to the same constant at the same scaling regardless of the set $\tilde{\mathcal{X}}$ so long as it has positive measure.

\section{Local Power}\label{power_sec}

This section derives asymptotic approximations to power functions by
considering the power of these tests under sequences of alternative
parameter values that approach the boundary of the identified set.
While we consider a single underlying distribution $P$ and sequence of local alternatives, \citet{armstrong_weighted_2014} shows that the test has power approaching one uniformly over certain classes of underlying distributions and parameters that are bounded away from the identified set by a sequence that approaches zero at the same rate \citep[technically, the results in][apply to a slightly different test where the truncation is done in a different way, but the results can be shown to apply to the version of the test considered here as well]{armstrong_weighted_2014}.  \citet{armstrong_choice_2014} shows that, under additional regularity conditions, several other tests considered in the literature perform strictly worse under the alternatives considered in this section, and in the uniform sense described above.
Appendix \ref{power_comp_sec} gives a more detailed description of these results and power comparisons with other tests in the literature.

Consider a parameter value $\theta_0$ on the boundary of the
identified set, and a sequence of local alternatives given by
$\theta_n=\theta_0+a r_n$ for some vector $a\in\mathbb{R}^{d_\theta}$
and some sequence of scalars $r_n\to 0$.
We impose the following conditions
\citep[see][for verification in several examples of a set of conditions that imply Assumption \ref{local_alt_assump}]{armstrong_weighted_2014}.

\begin{assumption}\label{local_alt_assump}
~
\begin{itemize}
\item[a.)] $\bar m(\theta,x)$ is differentiable in $\theta$ with derivative $\bar m_{\theta}(\theta,x)$ that is continuous as a function of $\theta$ uniformly in $(\theta,x)$.
\item[b.)] For some $\gamma$, $C$, $j$ and $x_0\in\mathcal{X}$, we
  have $\bar m_j(\theta_0,x_0)=0$ and, for all $x$ in a neighborhood
  of $x_0$,
\begin{align*}
|\bar m_j(\theta_0,x)-\bar m_j(\theta_0,x_0)|\le C\|x-x_0\|^\gamma.
\end{align*}
\end{itemize}
\end{assumption}

Part (b) of Assumption \ref{local_alt_assump} is a smoothness
condition on the conditional mean under $\theta_0$.  If $\bar
m_j(\theta_0,x_0)=0$ for some $x_0$, part (b) will hold with $\gamma=1$
if $\bar m_j(\theta_0,x)$ has a continuous first derivative in $x$, and
it will hold with $\gamma=2$ if $\bar m_j(\theta_0,x_0)=0$ has a
continuous second derivative in $x$ and $x_0$ is on the interior of
$\mathcal{X}$.

The following theorem gives local power results for sequences of local
alternatives.  To state the results, let $C(\cdot)$ be any bounded function on
the unit sphere such that Assumption \ref{local_alt_assump} holds with
$C$ replaced by $C((x-x_0)/\|x-x_0\|)$.  We can always take this
function to be a constant function under Assumption
\ref{local_alt_assump}, but, using this notation, we can state power
results that are more precise.

\begin{theorem}\label{local_alt_thm}
Suppose that %
Assumption \ref{local_alt_assump} holds for $\theta_0$
and that Assumption \ref{asym_dist_assump} holds with the constants in part (a) uniform over a neighborhood of $\theta_0$.
Let $\theta_n=\theta_0+a r_n$ for some $a\in\mathbb{R}^{d_\theta}$ and
a sequence of scalars $r_n\to 0$.
Suppose that, for some index $j$ such that part (b) of Assumption \ref{local_alt_assump} holds for $j$,
\begin{align*}
&\liminf r_n
\left(\frac{n}{2\log t_n^{-d_X}}\right)^{\gamma/(d_X+2\gamma)}  \\
&  > %
-\left\{\inf_{s,t} \frac{f(x_0)^{1/2}\int_{u\in\mathcal{U}, s<u<s+t}
  \left\{[\bar m_{\theta,j}(\theta_0,x_0) a]
  +C\left(\frac{u}{\|u\|}\right)\|u\|^\gamma\right\}\, du}
{\Sigma_{jj}^{1/2}(x_0)\text{vol}\{u\in\mathcal{U}|
  s<u<s+t\}^{1/2}}\right\}^{-\gamma/(d_X/2+\gamma)}
\end{align*}
where $\mathcal{U}=\cup_{k=1}^\infty (\mathcal{X}-x_0)/r_k$,
$\Sigma_{jj}(x)=var(m_j(W_i,\theta)|X_i=x)$
and the right hand side is taken be infinity if the infimum in the
brackets is zero.
Then, if $t_n<\eta (n/\log
n)^{-1/(d_X+2\gamma)}$ for small enough $\eta$, we will have
\begin{align*}
P(S_n(\theta_n)>\hat q_{1-\alpha})\to 1.
\end{align*}
\end{theorem}

Theorem \ref{local_alt_thm} states that, if $r_n$ is given by some constant $K$ times $[2(\log t_n^{-d_X})/n]^{\gamma/(d_X+2\gamma)}$, then the power of the test will converge to one so long as $K$ is strictly greater than the right hand side of the first display in the theorem.
If $\bar m_{\theta_j}(\theta_0,x_0) a$ is strictly negative, which
will typically be the case as long as $\theta_n$ is outside of the
identified set, then this result shows that the power of the test
approaches one as long as $\theta_n$ approaches $\theta_0$ at a
$(n/\log n)^{\gamma/(d_X+2\gamma)}$ rate with a large enough scaling.
This corresponds to the fastest rate
among available procedures even if
$\gamma$ were known, and corresponds to the optimal rate for certain related nonparametric testing problems
(see Appendix \ref{power_comp_sec} for details).
Theorem \ref{local_alt_thm} shows that our test is adaptive in the
sense that it achieves this rate simultaneously for all $\gamma$
without prior knowledge of $\gamma$.  Taking $t_n$ to be a $\log n$ term times $n^{-1/d_X}$, the condition that $t_n<\eta(n/\log n)^{-1/(d_X+2\gamma)}$ will be satisfied regardless of $\gamma$.  Another possibility is to take the smallest value of $\gamma$ that the researcher thinks is likely, and to choose a value of $t_n$ that is optimal for a particular data generating process and sequence of alternatives with that value of $\gamma$.  Theorem \ref{local_alt_thm} shows that this approach will achieve the optimal rate even if $\gamma$ is larger than the value used to choose $t_n$.

Note also that the conditions on $r_n$ depend only on $t_n$ and $\gamma$ and not on the volume of the set $\mathcal{X}$ or the number of moments $d_Y$, even though both of these quantities enter the critical value $\hat q_{1-\alpha}$.  This comes from the fact that $\hat q_{1-\alpha}/[(2\log t_n^{-d_X})^{1/2}/n^{1/2}]\to 1$ regardless of $\mathcal{X}$ and $d_Y$, as discussed in Section \ref{mom_sel_sec}.  Thus, if Theorem \ref{local_alt_thm} applies to show consistency against some sequence of alternatives given by $\theta_0+ar_n$ with a given set of moment functions $m_j(W_i,\theta)$ for $j=1,\ldots,d_Y$, it will also apply if one adds additional moment functions $m_k(W_i,\theta)$ for $k=d_Y+1,\ldots, d_Y+\ell$ for some $\ell>0$, even if these moment functions do not contribute to power.  We further explore the issue of irrelevant moments in one of our monte carlo designs in Section \ref{monte_carlo_sec}.

\section{Monte Carlo}\label{monte_carlo_sec}

We perform monte carlos with several designs based on a median regression model with potentially endogenously missing data.  We consider a missing data model where the conditional median of $W_i^*$ given $X_i$ is given by $q_{1/2}(W_i^*|X_i)=\theta_1+\theta_2 X_i$, and $W_i^*$ is missing for some observations.  Letting $W_i^H=W_i^*$ when $W_i^*$ is observed and $\infty$ otherwise, this leads to the conditional moment inequality $E[I(\theta_1+\theta_2 X_i\le W_i^H)-1/2|X_i]\ge 0$ a.s. (In practice, one would form another inequality based on a lower bound for $W_i^*$ of $-\infty$ when $W_i^*$ is not observed, but we focus on a single moment inequality here for simplicity.  We explore the consequences of including other moments later in this section.)

In each design, we simulate the data from a median
regression given by
$W_i^*=\theta_1^*+\theta_2^* X_i+u$
for some $(\theta_1^*,\theta_2^*)$ where $u\sim \text{unif}(-1,1)$ and $X_i\sim \text{unif}(0,1)$.  We then set $W_i^*$ to be missing with probability
$p(X_i)$ independently of $W_i^*$ for some function $p(x)$ (note that,
while we generate the data using a parameter value that satisfies
missingness at random, the test is designed to give confidence regions
that are robust to the failure of this assumption).  We consider 3 designs with $\theta_1^*=\theta_2^*=0$ and $p(x)$ given as follows for each design:
\begin{align*}
\begin{array}{ll}
\text{Design 1:} & p(x)=.1  \\
\text{Design 2:} & p(x)=.02+2\cdot .98\cdot |x-.5|  \\
\text{Design 3:} & p(x)=.02+4\cdot .98\cdot (x-.5)^2.
\end{array}
\end{align*}
Design 1 corresponds to a flat conditional mean, while Designs 2 and 3 correspond to $\gamma=1$ and $\gamma=2$ respectively in Assumption \ref{local_alt_assump}.  For each design, we consider the sample sizes $n=100,500,1000$ and the truncation parameters $t_n=n^{-1/5},n^{-1/3},n^{-1/2}$ for each sample size.  Note that $n^{-1/3}$ is the optimal rate for $t_n$ for Design 2 and $n^{-1/5}$ is the optimal rate for $t_n$ for Design 3, while $t_n=n^{-1/2}$ is smaller than optimal for all three designs, but still achieves the optimal rate for local alternatives by Theorem \ref{local_alt_thm}.

For each design, we test several parameter values with $\theta_2$ fixed at $0$ and $\theta_1$ varying.  For a given design, let $\overline \theta_1$ be the largest value of $\theta_1$ such that $(\theta_1,0)$ is in the identified set.  First, to examine the finite sample size of the test based directly on the asymptotic distribution, we report monte carlo estimates of the false rejection probability under $(\overline\theta_1,0)$ and Design 1, which corresponds to a least favorable null distribution with the conditional moment inequality equal to zero for all $x$.  This gives an idea of the worst (most liberal) size distortions one can expect from tests based on critical values calculated directly from the asymptotic distribution (at least, in situations similar to the median regressions with potentially endogenously missing or censored data considered here).

Table \ref{false_rejection_prob_table} reports these results.  We note that size distortions are generally minimal, except for the smaller sample sizes with the largest value of $t_n$, particularly with nominal size $\alpha=.1$.  As one might expect from the methods used in the derivation of the asymptotic distribution, which rely on tail approximations, the asymptotic approximation performs better for the smaller value of the nominal size $\alpha$.  The fact that size distortions are more severe with the larger $t_n=n^{-1/5}$ is likely a reflection of the fact that, for a fixed nominal size $\alpha$, the asymptotic approximations depend on $t_n$ being small relative to the support of $X_i$.  In contrast, size distortions are minimal for $t_n=n^{-1/3}$ for most cases considered here.

Next, we examine the power of our test.  We report monte carlo estimates of the power of our test for each design and parameters given by $(\overline\theta_1+a,0)$ for $a=.1,.2,.3,.4,.5$.  To ensure that power is not driven by false rejection under the null, we use critical values based on monte carlo estimates of the finite sample exact least favorable distribution.  We report power results for level $.05$ tests.  Tables \ref{power_table_d1}, \ref{power_table_d2} and \ref{power_table_d3} report the results.  As expected, moving away from the identified set by a given amount generally leads to more power under the designs with smoother conditional means.  In addition, the finding that the choice of the truncation parameter $t_n$ doesn't matter much as long as it is small enough appears to be borne out in the monte carlos (e.g. for Design 2, $t_n$ proportional to $n^{-1/3}$ is optimal, and this value of $t_n$ performs best, but choosing $t_n=n^{-1/2}$ gives close to the same power, while $t_n=n^{-1/5}$ gives much worse power).

Tables \ref{power_table_d1}, \ref{power_table_d2} and \ref{power_table_d3}
also include
the monte carlo rejection probability at the null value where $a=0$ for each design.  For Design 1, this corresponds to the least favorable null, so that the rejection probability is equal to size (up to ties in the monte carlo draws).  For Designs 2 and 3, the rejection probability at the boundary of the null ($a=0$) is strictly less than the size.  This is expected as well, since the conditional moment inequality is not binding for all values of $x$.

As discussed in Section \ref{mom_sel_sec}, the critical value will be conservative if one has prior knowledge of values of $x$ and $j$ such that the conditional mean $\bar m_j(\theta,x)$ is positive and large enough in magnitude so that these values of $x$ and $j$ do not affect the sampling distribution of the test statistic.  To examine this issue using monte carlo data, we consider additional designs where, in addition to observing $X_i$ and $W_i^H$ following one of the data generating processes described above, we observe an additional variable $\tilde W_i^H$, which is known to satisfy $E[I(\theta_1+\theta_2 X_i\le \tilde W_i^H)-1/2|X_i]\ge 0$ a.s., but where $\tilde W_i^H$ is large enough in magnitude that this inequality does not affect the sampling distribution of the test statistic under the parameter values considered in the monte carlos.\footnote{We thank an anonymous referee for suggesting these additional monte carlo designs}

In particular, we assume that $\tilde W_i^H$ is independent of $X_i$ and $W_i^H$
with $\tilde W_i^H\ge \theta_1+\theta_2 X_i$ a.s. for all values of $(\theta_1,\theta_2)$ considered in the monte carlos.
We use $\max\{T_{n,1},T_{n,2}\}$ where $T_{n,1}$ and $T_{n,2}$ are the test statistics formed using each of these moment inequalities, with the same tuning parameters.
We then generate critical values based on the distribution where both equalities bind almost surely.  This corresponds to $m_1(W_i,\theta)$ and $m_2(W_i,\theta)$ being distributed $\text{Bernoulli}(1/2)$ independently of each other and the $X_i$'s.  Note that, $T_{n,1}$ and $T_{n,2}$ are iid under this distribution, so we can compute the critical value as the $(1-\alpha)^{1/2}$ quantile of the distribution of $T_{n,1}$.  Thus, performing the monte carlo analysis with these designs is equivalent to performing the original monte carlo analysis with $\alpha$ replaced by $1-(1-\alpha)^{1/2}$.

Results are reported in Tables \ref{power_table_em_d1}, \ref{power_table_em_d2} and \ref{power_table_em_d3}.  While there is some decrease in power relative to the tests where the irrelevant moment is not included, this decrease in power is relatively small in most cases.

Finally, we consider tests in \citet{chetverikov_adaptive_2012}, which are based on test statistics similar to the one proposed in the present paper, but which include versions of the test that use moment selection procedures.
We provide details of the tests in Appendix \ref{additional_mc_sec}.  Tables \ref{power_table_ch_d1}, \ref{power_table_ch_d2} and \ref{power_table_ch_d3} report the results.
\citet{armstrong_choice_2014}
reports the
results of a monte carlo analysis of some other tests under the same designs.

The monte carlo results for these tests include two ways of forming the critical value.
The plug-in asymptotic (PIA) critical value uses a bootstrap estimate of the distribution of the test statistic when all moments bind, while the generalized moment selection (GMS) critical value uses a pre-test to determine which moments are close to binding.
The performance of these tests is generally similar to the tests considered here (depending on the tuning parameter $t_n$).
The power improvement from using critical values based on moment selection (RMS) is small, which confirms the prediction in Section \ref{mom_sel_sec} that test statistics of this form (both the test proposed by \citet{chetverikov_adaptive_2012} and the one proposed in the present paper) are not very sensitive to moment selection.  This also mirrors the findings in the monte carlo designs reported by \citet{chetverikov_adaptive_2012}, in which the test statistic proposed in that paper has similar power regardless of whether PIA or RMS critical values are used.

\section{Empirical Illustration}\label{application_sec}

We apply our methods to a median regression model with endogenously censored and missing data, using data from the Health and Retirement Study.  The setup follows Section 9 of \citet{armstrong_asymptotically_2011}, but we repeat it here for convenience.  Letting $X_i$ and $W_i^*$ be yearly income and prescription drug expenditures for participant $i$ respectively, we posit the model
\begin{align}
q_{1/2}(W_i^*|X_i)=\theta_1+\theta_2 X_i
\end{align}
where $q_{1/2}(W_i^*|X_i)$ is the median of $W_i^*$ conditional on $X_i$.

In this survey, participants who did not report a point value for prescription drug expenditures were given a series of brackets for this variable, resulting in interval censoring for a portion of the observations, and some observations with a completely missing outcome variable.  In other words, we do not observe $W_i^*$, but only observe a random interval $[W_i^L,W_i^H]$ known to contain $W_i^*$.  The data is censored in a way that is likely to violate a missingness at random or censoring at random assumption: the variable is censored only for those who do not recall how much they spent, and it is likely that remembering how much one spent is correlated with the level of spending itself.

This endogenous censoring problem makes it impossible to estimate $(\theta_1,\theta_2)$ consistently in general.  We construct bounds using the conditional moment inequalities
\begin{align}\label{int_reg_ineq}
E[m(X_i,W_i^L,W_i^H,\theta)|X_i]
\equiv E\left[\left.\begin{array}{c}I(\theta_1+\theta_2 X_i\le W_i^H)-1/2  \\
1/2-I(\theta_1+\theta_2 X_i\le W_i^L)
\end{array} \right|X_i\right]
\ge 0   \,\,\,\,\,\,\,\,\,  \text{a.s.}
\end{align}
We test (\ref{int_reg_ineq}) at the .05 level using our methods for each value of $(\theta_1,\theta_2)$, and report a $95\%$ confidence region that inverts these tests.  The resulting confidence region contains the true parameter value with probability at least $.95$.

We restrict our sample to the 1996 wave of the survey and women with no more than \$15,000 of yearly income who report using prescription medications.  The data set also contains observations with a censored covariate (income), but, for illustrative purposes, we focus on endogenous censoring of the outcome variable and throw away observations where income is missing or censored (this is valid if remembering prescription drug expenditures is not correlated with income, but may be correlated with spending itself).  Our data set has 636 observations, of which 54 have an interval censored outcome variable, and an additional 7 have a completely missing outcome variable.  See \citet{armstrong_asymptotically_2011} for additional details about the data set.
For the truncation parameter $t_n$, we use
$n^{-1/3}\cdot (\max_{1\le i\le n} X_i-\min_{1\le i\le n} X_i)$.
The $n^{-1/3}$ scaling results in a test statistic that is rate adaptive to smoothness between Lipschitz continuity and 2 derivatives of the conditional truncation probabilities (a smaller value could be used to adapt to a less smooth data generating process).
For the critical value for our test, we use the analytically computed critical value defined in (\ref{cval_eq}).

Figure \ref{cr95_asym_dist_fig} shows the resulting confidence region.  For comparison, Figures \ref{cr95_est_fig} and \ref{cr95_cons_fig} show confidence regions using
the test statistic $\left|\inf_{I(s,t)\subseteq \hat{\mathcal{X}}} E_n m(W_i,\theta)I(s<X_i<s+t)\right|_{-}$, which is similar to the statistic used in this paper, but does not weight moments by the inverse of their standard deviation.
The constant weighting used by this statistic falls under the conditions of \citet{andrews_inference_2013} and \citet{kim_kyoo_il_set_2008}.
These figures are
taken directly from \citet{armstrong_asymptotically_2011}.  Figure \ref{cr95_est_fig} uses this statistic along with a critical value proposed in \citet{armstrong_asymptotically_2011}, while Figure \ref{cr95_cons_fig} uses a critical value that is more conservative, but valid under weaker conditions.
The test considered in this paper can be thought of as introducing an optimal weighting to the \citet{andrews_inference_2013} statistic.
This improves the rate for local alternatives from $n^{-\gamma/(2d_X+2\gamma)}$ to the $(n/\log n)^{-\gamma/(d_X+2\gamma)}$ rate obtained in Theorem \ref{local_alt_thm} in the set identified case, while worsening the rate by a $\log n$ term in the point identified case.  The \citet{armstrong_asymptotically_2011} test yields a slightly better improvement in power in certain situations, but is not robust to failure of certain smoothness conditions.
We also report confidence regions for each component of $(\theta_1,\theta_2)$, formed by projecting the confidence region onto each component.  Table \ref{ci_table} reports these confidence intervals, along with the corresponding confidence intervals formed using other methods reproduced from \citet{armstrong_asymptotically_2011} for convenience.

The slope parameter, $\theta_2$, gives the median increase in yearly prescription drug spending associated with an increase in income.  Thus, according to the results using the test proposed in this paper, a 95\% confidence interval puts the median increase in prescription drug expenditures associated with a \$1,000 in income between \$5.30 and \$32.00.  It is worth making a few notes in comparing this with the confidence regions using the unweighted statistic.  As predicted by the asymptotic power results, the confidence region for the slope parameter is tighter than the one obtained using an unweighted test statistic with a critical value formed using subsampling with a conservative rate.  The unweighted statistic gives a better lower bound for the slope parameter when subsampling with an estimated rate is used to form the critical value, but this test is less robust in the sense that it relies on additional smoothness conditions.

Comparing the joint confidence regions for $(\theta_1,\theta_2)$, we see that the tests based on unweighted statistics with subsampling based critical values lead to disconnected regions of rejected and accepted parameter values.  While the test based on a conservative rate (Figure \ref{cr95_cons_fig}) has only a small island of rejected parameter values in the confidence region, the test based on an estimated rate proposed in \citet{armstrong_asymptotically_2011} (Figure \ref{cr95_est_fig}) leads to numerous isolated areas in the confidence region.  In contrast, our test leads to a connected confidence region.  A likely explanation for this phenomenon is that the subsampling based confidence regions use critical values that implicitly estimate where the data generating process is in the null space.  This leads to disconnected confidence regions when, as the parameter moves in some direction, the test first begins to reject as the test statistic increases, but then fails to reject when the critical value increases as well.  In contrast, our test uses a least favorable critical value, so the test always moves from acceptance to rejection as the test statistic increases.

\section{Conclusion}\label{conclusion_sec}

This paper considers inference in conditional moment inequality models using a multiscale statistic.  The asymptotic distribution of our test statistic is derived, and the results are used to obtain feasible critical values.
The test obtains
certain optimal rates for power against local alternatives adaptively, and is the only feasible test available that does so for the best possible range of smoothness classes.  Our results also have implications for the effect of moment selection procedures on power, and our test has the additional advantage of being adaptive without requiring such tests.
An empirical application to a regression model with endogenous censoring and missing data illustrates the power improvement from the test.

\appendix

\section{Uniformity in the Underlying Distribution}\label{uniformity_sec}

We prove a stronger version of Theorem \ref{asym_dist_thm_randx} that holds uniformly in certain classes $\mathcal{P}$ of underlying distributions for which Assumption \ref{asym_dist_assump} holds uniformly over $P\in\mathcal{P}$.  To state and prove this result, we introduce some notation for indexing certain quantities by the underlying distribution $P$.  We use the notation $E_P$ to denote expectation with respect to the probability distribution $P$, and use similar notation for conditional expectations and conditional and unconditional variances, covariances and correlations.  We make explicit the dependence of the identified set on $P$ and define $\Theta_0(P)=\{\theta\in\Theta|E_P[m(W_i,\theta)|X_i]\ge 0 \text{ $a.s.$}\}$.

In the following theorem, the conditional distribution (including the conditional mean) of $m(W_i,\theta)$ given $X_i=x$ is allowed to vary over $\mathcal{P}$.
In particular, since no conditions are placed on the conditional mean of distributions in $\mathcal{P}$, the result shows that tests based on this asymptotic distribution result control the asymptotic size uniformly over distributions for which the conditional mean can be nonsmooth in arbitrary ways, although there are some mild continuity assumptions on the conditional variance.  We do, however, impose the same distribution of $X_i$ for all $P\in\mathcal{P}$.  This is mostly to avoid introducing additional notation in the proof, and could be relaxed (although the volume of the support would have to be bounded away from zero and the boundary would have to be uniformly well behaved in some sense).

\begin{theorem}\label{asym_dist_thm_randx_unif_p}
Let $\hat c_n$, $a(\hat c_n)$ and $b(\hat c_n)$ be defined as in Theorem \ref{asym_dist_thm_randx}.  Suppose that Assumption \ref{asym_dist_assump} holds for the same constants in part (a) for all $P\in\mathcal{P}$ and with the continuity in part (ii) of part (a) uniform over $P\in\mathcal{P}$.  Then, for any vector $r\in\mathbb{R}^{d_Y}$,
\begin{align*}
\liminf_n\inf_{P\in\mathcal{P},\theta_0\in\Theta_0(P)} P(a(\hat c_n) T_n(\theta_0)-b(\hat c_n)\le r)
\ge P(Z\le r)
\end{align*}
where $Z$ is a $d_Y$ dimensional vector of independent standard type $I$ extreme value random variables.
If, in addition, %
$E_P[m_j(W_i,\theta_0)|X_i=x]=0$ for all $j$ and $x$
for all $P\in\mathcal{P}$ for some $\theta_0$, then, for this $\theta_0$,
\begin{align*}
a(\hat c_n)T_n-b(\hat c_n)\stackrel{d}{\to} Z
\end{align*}
uniformly over $P\in\mathcal{P}$.
\end{theorem}

The second display in Theorem \ref{asym_dist_thm_randx_unif_p} shows that, for certain classes of underlying distributions where the conditional moment inequalities all bind for all values of $x$, our test is uniformly asymptotically similar.  While this typically only holds for very restricted classes (e.g. classes of distributions for the missing data model in Section \ref{application_sec} where the probability of missingness is zero), we include it here for completeness.
Note that the second display of Theorem \ref{asym_dist_thm_randx_unif_p} is stronger than necessary for the test to have asymptotic size $\alpha$.  Since size is defined as the supremum of the rejection probability over $\mathcal{P}$ with $\theta_0\in\Theta_0(P)$, the test will have asymptotic size $\alpha$ so long as
the first display in Theorem \ref{asym_dist_thm_randx_unif_p} holds
and
there exists a $P^*\in\mathcal{P}$ with $\theta_0\in\Theta_0(P^*)$ and
$E_{P^*}[m_j(W_i,\theta_0)|X_i=x]=0$ for all $x$ and $j$.  This follows from Theorem \ref{asym_dist_thm_randx_unif_p} along with the second display of Theorem \ref{asym_dist_thm_randx}.
Thus, in the missing data model in Section \ref{application_sec}, the test will have asymptotic size $\alpha$ over any class $\mathcal{P}$ that satisfies certain regularity conditions so long as it contains a distribution with no missingness.

We now comment briefly on the conditions on $\mathcal{P}$ and their relation to conditions used in other results in the literature.  First, note that the primary concern for uniform-in-$P$ asymptotics in the moment inequality literature is moment selection, which leads to concerns that a procedure may not be uniform in classes $\mathcal{P}$ where the inequality may be close to, but not quite, binding.  Since our procedure does not use moment selection, one might have less reason to be concerned and, indeed, the class $\mathcal{P}$ in Theorem \ref{asym_dist_thm_randx_unif_p} allows for such cases since it does not place any conditions on the conditional mean $E_P(m(W_i,\theta)|X_i=x)$.  Other tests in the literature have also been shown to be robust to classes of underlying distributions that place mild conditions or no conditions on the conditional mean, including \citet{andrews_inference_2013}, \citet{lee_testing_2013} and \citet{chetverikov_adaptive_2012} (the latter paper assumes some smoothness for the conditional mean, but allows for the cases where moments are ``nearly binding,'' which are the main concern in this literature).  \citet{chernozhukov_intersection_2013} place smoothness assumptions on the conditional mean, which is necessary in the case where higher order kernels or sieves are used, but could be relaxed for the case of a positive kernel.
Regarding the conditional variance, we assume continuity, as does \citet{chetverikov_adaptive_2012}.  Note that \citet{andrews_inference_2013} obtain uniformity in classes of distributions for which the set of covariance kernels for a certain process is compact, which may place some conditions on the conditional variance.  Regarding our exponential moment condition, \citet{chernozhukov_intersection_2013} also use strong moment assumptions in certain cases, while \citet{andrews_inference_2013} and \citet{chetverikov_adaptive_2012} only require polynomial moments.  We use the exponential moment condition to verify conditions for moderate deviations approximations, which allow us to take $t_n\to 0$ at the best possible rate (note that \citet{chetverikov_adaptive_2012} places stronger conditions on the rate at which the analogue of $t_n$ in that paper approaches zero, which preclude adaptivity in certain settings; however, the conditions in that paper, as well as ours, could be changed to trade off conditions on $t_n$ and moment conditions in other ways).

For completeness, we also include the following theorem, which states that the tests proposed in this paper control the size uniformly over classes of distributions that satisfy the conditions of the above theorem.

\begin{theorem}\label{test_size_unif_thm}
For any class $\mathcal{P}$ of distributions satisfying the conditions of Theorem \ref{asym_dist_thm_randx_unif_p},
\begin{align*}
\limsup_n\sup_{P\in\mathcal{P},\theta_0\in\Theta_0(P)} P(S_n(\theta_0)>\hat q_{1-\alpha})\le \alpha.
\end{align*}
\end{theorem}

Theorem \ref{test_size_unif_thm} follows immediately from Theorem \ref{asym_dist_thm_randx_unif_p}.
We prove Theorem \ref{asym_dist_thm_randx_unif_p} in the next appendix.

\section{Proof of Theorem \ref{asym_dist_thm_randx_unif_p}}

We first prove a version of Theorem \ref{asym_dist_thm_randx_unif_p} where the $X_i$s are deterministic and $\hat\sigma^2$ is replaced by a certain sample average of conditional variances defined below.  The result then follows from showing that the conditions of this result hold almost surely conditional on $\{X_i\}_{i=1}^n$,
and that replacing the sample average of conditional variances with $\hat\sigma^2$ does not change the test statistic too much.

Throughout this section, we fix $\theta$ and let $Y_i=m(W_i,\theta)$, and drop the $\theta$ notation elsewhere such as in the definition of $\hat\sigma_{n,j}(s,t,\theta)$.  We prove the following result with $\{X_i\}_{i=1}^n$ replaced by a deterministic sequence $\{x_i\}_{i=1}^n$.
We consider %
a set $\mathcal{P}$ determining the probability distribtuion of $Y_i$ for a given $x_i$.

Let ${\cal F} = \{ F_{x,P}: x \in {\cal X}, P\in\mathcal{P} \}$ be a family of $d_Y$-dimensional distribution functions,
with $\cX$ a compact, Jordan measurable subset of $\mathbb{R}^{d_X}$ such that vol($\cX) > 0$, that is
it has positive $d_X$ dimensional volume.
Consider $(x_1, Y_1), (x_2, Y_2), \ldots$ with 
$x_i$ deterministic and $Y_i \sim F_{x_i}$ independent. Define 
$\mu_P(x) = E_{x,P} Y_i$ and $\Sigma_P(x) = {\rm Cov}_{x,P} Y_i$, where the
subscript $x,P$ denotes with respect to $Y_i \sim F_{x,P}$.
We use the notation $z_{i,j}$ to denote the $j$th coordinate of the $i$th observation or element in a sequence $\{z_i\}$.
Let $I(s,t)= \prod_{j=1}^{d_t} [s_j,s_j+t_j)$.
We abuse notation slightly and define $\text{vol}(t) = \prod_{j=1}^{d_t} t_j$ for a vector $t$.
Let $J_n(s,t) = \{ i: 1 \leq i \leq n,
x_i \in I(s,t) \}$. %
We consider the following regularity conditions.

\begin{assumption}%
\label{assump_a1}
~
\begin{enumerate}
\item[a.)] There exists $\lambda>0$ and $M_\lambda < \infty$ such that
$$E_{x,P} (e^{\lambda |Y_{i,j}|}) \leq M_\lambda \mbox{ for all } x \in \mathcal{X}, 1 \leq j \leq d_Y, P\in\mathcal{P}.
$$
Hence the characteristic function of $Y_{i,j}$ is analytic on $(-\lambda,\lambda)$ for all $j$ and
$Y_i \sim F_{x,P}$, $x \in \cal X$, $P\in\mathcal{P}$.

\item[b.)] $\sigma_{j,P}(x) \equiv \Sigma_{jj,P}^{1/2}(x)$ is continuous and positive on $\mathcal{X}$ for all
$1 \leq j \leq d_Y$ uniformly over $P\in\mathcal{P}$.

\item[c.)] (for $d_Y>1$): 
$$\rho \equiv 
\sup_{P\in\mathcal{P}}\sup_{i \neq j} \sup_{x \in \cal X} 
\frac{\Sigma_{ij,P}(x)}{\sigma_{i,P}(x) \sigma_{j,P}(x)} < 1. 
$$
\end{enumerate}
\end{assumption}

\begin{assumption}%
\label{assump_a2}
There exists a continuous, positive and bounded 
density function
$f$ on $\mathcal{X}$ and a sequence $t_n \rightarrow 0$ such that
\begin{enumerate}
\item[a.)] $nt_n^{d_X} | \log t_n|^{-4}
\rightarrow \infty$,

\item[b.)] for any $\delta>0$, $\# J_n(s,t) \sim n \int_{I(s,t)} f(x) dx$ 
uniformly over
$I(s,t) \subseteq
\cal X$ such that $\text{vol}(t) \geq \delta t_n^{d_X}/|\log t_n|^2$. 
\end{enumerate}
\end{assumption}

\noindent Define $\sigma_{n,j}(s,t) = \{ \sum_{i \in J_n(s,t)} [\sigma_{j,P}(x)]^2 \}^{1/2}$ and let
$$\tilde T_{n,j} = - \inf_{I(s,t) \subseteq {\cal X}, t \geq t_n {\bf 1}} \sum_{i \in J_n(s,t)}
Y_{i,j} \Big/ [\sqrt{n}\sigma_{n,j}(s,t)]
$$
(we suppress the dependence of $\sigma_{n,j}(s,t)$ and $\tilde T_{n,j}$ on $P$ for notational convenience).

\begin{theorem}\label{asym_dist_thm}
Suppose that $\mu_P(x) \geq 0$ for all $x \in \cal X$, $P\in\mathcal{P}$ and that Assumptions \ref{assump_a1} and \ref{assump_a2} hold. 
Let $a_n=(2 n\log t_n^{-d_X})^{1/2}$ and
$b_n=2\log t_n^{-d_X}+(2d_X-\frac{1}{2})\log\log t_n^{-d_X}
  -\log [2\sqrt{\pi}/\text{vol}(\mathcal{X})]$.
Then, for any vector $r\in\mathbb{R}^{d_Y}$,
\begin{align*}
\liminf_{n\to\infty}\inf_{P\in\mathcal{P}}
  P\left(a_n \tilde T_n-b_n\le r\right)\ge P(Z\le r)
\end{align*}
where $Z$ is a $d_Y$ dimensional vector of independent standard type I extreme value random variables.  If, in addition, $\mu_P(x) = 0$ for all $x \in \cal X$, $P\in\mathcal{P}$, then
\begin{align*}
\lim_{n\to\infty}\sup_{P\in\mathcal{P}}
  \left|P\left(a_n \tilde T_n-b_n \le r\right)- P(Z\le r)\right|=0.
\end{align*}
\end{theorem}

The result follows from this and the following lemmas.

\begin{lemma}\label{x_lemma}
Under Assumption \ref{asym_dist_assump}, part (b) of Assumption \ref{assump_a2} above holds for almost all sequences $\{X_i\}_{i=1}^n$.
\end{lemma}
\begin{proof}
We have
\begin{align*}
\frac{\# J_n(s,t)}{n \int_{I(s,t)} f(x) dx}-1
=\frac{E_n I(s<X_i<s+t)-E I(s<X_i<s+t)}{E I(s<X_i<s+t)}.
\end{align*}
This converges to one uniformly over $(s,t)$ with $vol(t)\ge K_n (\log n)/n$ for any sequence $K_n\to\infty$ by Theorem 37 in Chapter 2 of \citet{pollard_convergence_1984}, and the conditions $nt_n^{d_X}/|\log t_n|^4\to\infty$ and 
$vol(t) \geq \delta t_n^{d_X}/|\log t_n|^2$
guarantee that $vol(t)\ge \delta t_n^{d_X}/|\log t_n|^2\ge K_n n^{-1}|\log t_n|^4/|\log t_n|^2\ge K_n (\log n)/n$
for some $K_n\to \infty$.
\end{proof}

\begin{lemma}\label{vol_lemma}
Under Assumption \ref{asym_dist_assump}, $vol(\hat{\mathcal{X}})\stackrel{p}{\to} vol(\mathcal{X})$.
\end{lemma}
\begin{proof}
For a given $\varepsilon,\delta>0$, the following event will hold with probability approaching one:
for every point $\varepsilon k$ in the grid $(\varepsilon \mathbb{Z}^{d_X})\cap\mathcal{X}$,
at least one observation $X_i$ will have each component $X_{i,j}$ within $\delta$ of $\varepsilon k$.
Once this holds, the set $\varepsilon I((k_1+\delta,\ldots,k_{d_X}+\delta),(1-\delta,\ldots,1-\delta))$ will be contained in the convex hull of the $X_i$s for all $k$ such that
$\varepsilon I(k,{\bf 1})\subseteq \mathcal{X}$.
This gives a lower bound of
$(1-2\delta)^{d_X} vol(\cup_{\varepsilon I(k,{\bf 1})\subseteq \mathcal{X}} \varepsilon I(k,{\bf 1}))$
for the volume of the convex hull of the $X_i$s, which can be made arbitrarily close to $vol(\mathcal{X})$ by Jordan measurability.  The result follows from this and the upper bound $vol(\hat{\mathcal{X}})\le vol(\mathcal{X})$.
\end{proof}

\begin{lemma}\label{sigmahat_lemma}
Under Assumptions \ref{assump_a1} and \ref{assump_a2} (with the $X_i$'s treated as nonrandom),
$\sup_{s,s+t\in \hat{\mathcal{X}}, t\ge t_n} \frac{\sigma_{n,j}(s,t)}{\sqrt{n}\hat\sigma_{n,j}(s,t)}-1\le o_P(\log n)^{-1}$
uniformly over $P\in\mathcal{P}$
and, if $\bar m(\theta,x)=0$ for all x,
  $\sup_{s,s+t\in \hat{\mathcal{X}}, t\ge t_n} \left|\frac{\sigma_{n,j}(s,t)}{n\hat\sigma_{n,j}(s,t)}-1\right|=o_P(\log n)^{-1}$
uniformly over $P\in\mathcal{P}$.
\end{lemma}
\begin{proof}
First, note that, since $x\mapsto 1/x^2$ is decreasing and differentiable at one, it suffices to show that $\inf_{s,s+t\in\hat{\mathcal{X}}} \frac{n\hat\sigma_{n,j}^2(s,t)}{\sigma_{n,j}^2(s,t)}-1\ge -o_P(\log n)^{-1}$ and
$\sup_{s,s+t\in\hat{\mathcal{X}}} \left|\frac{n\hat\sigma_{n,j}^2(s,t)}{\sigma_{n,j}^2(s,t)}-1\right|
  = o_P(\log n)^{-1}$.
Note that
\begin{align*}
&\hat \sigma_{n,j}^2(s,t)-\sigma_{n,j}^2(s,t)/n
=\frac{1}{n}\sum_{i\in J_n(s,t)} Y_{i,j}^2-\left[\frac{1}{n}\sum_{i\in J_n(s,t)} Y_{i,j}\right]^2
-\frac{1}{n}\sum_{i\in J_n(s,t)} \sigma_{j,P}(x)^2  %
= I+II
\end{align*}
where
\begin{align*}
I\equiv \frac{1}{n}\sum_{i\in J_n(s,t)} \left(Y_{i,j}^2- E_{x_{i}} Y_{i,j}^2\right)
\end{align*}
and
\begin{align*}
II\equiv \frac{1}{n}\sum_{i\in J_n(s,t)} \left[E_{x_i} Y_{i,j}^2-\sigma_{j,P}(x)^2\right]
-\left[\frac{1}{n}\sum_{i\in J_n(s,t)} Y_{i,j}\right]^2
=\frac{1}{n}\sum_{i\in J_n(s,t)} [E_{x_{i}} Y_{i,j}]^2
-\left[\frac{1}{n}\sum_{i\in J_n(s,t)} Y_{i,j}\right]^2.
\end{align*}

We first bound $I/[\sigma^2_{n,j}(s,t)/n]$ where $I$ is given above.
Let $W_i=Y_{i,j}^2-E_{x_i,P}Y^2_{i,j}$.
Note that $\sigma_{n,j}^2(s,t)$ is bounded from below by a constant times $\# J_n(s,t)$ uniformly over $P\in\mathcal{P}$, so it suffices to consider
$\left(\sum_{i\in J_n(s,t)} W_i\right)/\# J_n(s,t)$.
For some sequence $K_n$, let $\tilde W_i=W_iI(|W_i|\le K_n)$ be a truncated version of $W_i$.  Note that, by Markov's inequality, for $\lambda>0$ given in Assumption \ref{assump_a1},
\begin{align*}
P\left(|W_i|> K\right)
\le E_{x_i,P}\exp(\lambda \sqrt{|W_i|}-\lambda\sqrt{K})
\end{align*}
so
\begin{align*}
P\left(|W_i|> K \text{ some $1\le i\le n$}\right)
\le n\exp(-\lambda\sqrt{K}) \sup_{x\in\mathcal{X},P\in\mathcal{P}} E_{x,P}\exp(\lambda \sqrt{|W_i|}),
\end{align*}
which goes to zero for any $K=K_n$ that increases faster than $(\log n)^2$.  To bound
$|E_{x_i,P} \tilde W_i|=|E_{x_i,P} \tilde W_i-E_{x_i,P} W_i|$, note that
\begin{align*}
\{E_{x_i,P}[|W_i| I(|W_i|>K)]\}^2
\le E_{x_i,P}(W_i^2) P(|W_i|>K)
\le C \exp(-\lambda \sqrt{K})
\end{align*}
for some constant $C$ that does not depend on $P$ or $x_i$.  Thus,
$|\sum_{i\in J_n(s,t)} E_{x_i,P} \tilde W_i|/\# J_n(s,t)\le [C \exp(-\lambda\sqrt{K_n})]^{1/2}$, which goes to zero at a polynomial rate for $K_n$ increasing faster than $(\log n)^2$, which is faster than the required $\log n$ rate.

Using the fact that the supremum over $(s,t)$ is determined by the maximum over no more than $n^{2d_X}$ possible deterministic configurations for $J_n(s,t)$, and that for any $\delta > 0$,
$\delta (\log n)^{-1} \ge \# J_n(s,t)^{-1/4}$ for large enough $n$,
\begin{align*}
&P\left(\sup_{s,s+t\in\hat{\mathcal{X}}, t\ge t_n} \left|\frac{\sum_{i\in J_n(s,t)} [\tilde W_i-E_{x_i,P}\tilde W_i]}{\# J_n(s,t)}\right|\ge \delta (\log n)^{-1}\right)  \\
&\le n^{2d_X}
\sup_{s,s+t\in\hat{\mathcal{X}}, t\ge t_n} P\left(\left|\frac{\sum_{i\in J_n(s,t)} [\tilde W_i-E_{x_i,P}\tilde W_i]}{\# J_n(s,t)}\right|\ge \# J_n(s,t)^{-1/4}\right).
\end{align*}
Now, using Bernstein's inequality, for $C$ a bound for the fourth moment of $Y_{i,j}$, the above display is bounded by
\begin{align*}
n^{2d_X} \sup_{s,s+t\in\hat{\mathcal{X}}, t\ge t_n} 2\exp\left(-\frac{[\# J_n(s,t) ^{3/4}]^2}{C \# J_n(s,t)+K_n [\# J_n(s,t)^{3/4}]/3}
  \right).
\end{align*}
Let $K_n$ be such that $K_n\le \# J_n(s,t)^{1/2}$ all $(s,t)$ and $K_n/(\log n)^2\to\infty$.
For large enough $n$, this gives a bound in the above display of
$n^{2d_X}\sup_{s,s+t\in\hat{\mathcal{X}}, t\ge t_n} \exp(-\# J_n(s,t)^{1/4})\to 0$.

As for $II$, we have
\begin{align*}
&II\ge \left[\frac{1}{n}\sum_{i\in J_n(s,t)} E_{x_{i},P} Y_{i,j}\right]^2
-\left[\frac{1}{n}\sum_{i\in J_n(s,t)} Y_{i,j}\right]^2  \\
&\ge - 2\left(\left|\frac{1}{n}\sum_{i\in J_n(s,t)} E_{x_{i},P} Y_{i,j}\right|
\vee\left|\frac{1}{n}\sum_{i\in J_n(s,t)} Y_{i,j}\right|\right) \left| \frac{1}{n}\sum_{i\in J_n(s,t)} (Y_{i,j}-E_{x_i,P} Y_{i,j}) \right|
\end{align*}
and similar methods show that the last line divided by $\sigma_{n,j}(s,t)/n$ converges to zero at a faster than $\log n$ rate uniformly over $(s,t)$ with $s,s+t\in\hat{\mathcal{X}}$, $t\ge t_n$.  If $\bar m_j(\theta,x)=0$ for all $x$, then
$II=-\left[\frac{1}{n}\sum_{i\in J_n(s,t)} Y_{i,j}\right]^2$,
and $\left[\frac{1}{n}\sum_{i\in J_n(s,t)} Y_{i,j}\right]^2/[\sigma_{n,j}(s,t)/n]$ 
also converges to zero
at a faster than $\log n$ rate 
uniformly over $(s,t)$ with 
$s,s+t\in\hat{\mathcal{X}}$, $t\ge t_n$ by similar arguments.
\end{proof}

\subsection{Proof of Theorem \ref{asym_dist_thm}}

We begin by proving the result in the case of a univariate outcome $Y_i=m(W_i,\theta)$.  Section \ref{multi_y_sec} generalizes the result to the case of multivariate $Y_i$.

To simplify notation, %
we let $d=d_X$ and we omit the subscript $P$ when dealing with expectations and other quantities that depend on the underlying distribution $P$. 
Let $Z_i=\mu(x_i)-Y_i$ and let $B_n = \{ (s,t): I(s,t) \subseteq {\cal X}, t \geq t_n {\bf 1} \}$.
Define 
\begin{equation} \label{Hn}
\mathbb{H}_n(s,t) = \frac{\sum_{i \in J_n(s,t)} Z_i}{\sigma_n(s,t)} \mbox{ for } (s,t) \in B_n.
\end{equation}
Let $c=\frac{(b_n+\zeta)\sqrt{n}}{a_n}$. Note in particular that
\begin{eqnarray} \label{c2}
& & t_n^{-d} (c^2/2)^{2d-\frac{1}{2}} e^{-c^2/2} \\
& \sim & t_n^{-d} (d |\log t_n|)^{2d-\frac{1}{2}} \exp \Big( -\frac{1}{2}
\Big\{ (2d | \log t_n|)^{1/2}+\frac{\log[(d | \log t_n|)^{2d-\frac{1}{2}} \text{vol}(\cX) e^\zeta
/2 \sqrt{\pi}]}{(2d | \log t_n|)^{1/2}} \Big\}^2 \Big) \nonumber \\
& \rightarrow & [2 \sqrt{\pi}/\text{vol}(\cX)] e^{-\zeta} \mbox{ as } n \rightarrow \infty. \nonumber
\end{eqnarray}
Theorem \ref{asym_dist_thm} in the case of univariate $Y$ follows from 
\begin{equation} \label{lP}
\lim_{n \rightarrow \infty}\sup_{P\in\mathcal{P}} \left|P \{ \sup_{(s,t) \in B_n} 
\mathbb{H}_n(s,t) \geq c \} - [1-\exp(-e^{-\zeta})]\right|\to 0 \mbox{ for all } \zeta \in \mathbb{R}.
\end{equation}

Consider a change-of-variables by defining $\mX_c$ such that
\begin{equation} \label{Xc}
\mX_c(-u,v) = \mathbb{H}_n(ut_n,(v-u)t_n) \mbox{ for } (ut_n,(v-u)t_n) \in B_n.
\end{equation}
The domain of $\mX_c$ is thus $D_c\equiv \{ (-u,v) \in (-t_n^{-1} \cX) \times t_n^{-1} \cX: v-u \geq
{\bf 1} \}$.
Note that $\mathbb{X}_c(-u,v)$ is a normalized sum over observations for which $x_i$ lies in the rectangle
$\{x|u t_n< x<v t_n\}$.  The change of variable and unusual notation are designed so that, for $a,b\ge 0$, the rectangle associated with $\mathbb{X}_c(-u+a,v+b)$ contains the rectangle associated with $\mathbb{X}_c(-u,v)$.  This helps with the verification of some of the conditions in \citet{chan_maxima_2006} involving positive increments of the process.
 
Let $\psi(z) = \frac{1}{z \sqrt{2 \pi}} e^{-z^2/2}$ and $\Delta_c = (2c^2)^{-1}$. 
Consider a restriction of $D_c$ to $D_L(=D_{c,L})\equiv \{ (-u,v) \in D_c: v-u \leq L {\bf 1} \}$
for some $L > 1$. Let
\begin{equation} \label{Dwstar}
D^*_w = \{ (-u,v) \in -I(w,|\log t_n|) \times I(w,|\log t_n|): {\bf 1} \leq (v-u) \leq 
L {\bf 1} \}.
\end{equation}
We will show that regularity conditions (C) and (A1)--(A5) in 
Corollary 2.7 of \citet{chan_maxima_2006} are satisfied 
uniformly on the domains $D_L$ and over $P\in\mathcal{P}$ and hence  
\begin{equation} \label{Psup}
q_{w,P} \equiv P \{ \sup_{(-u,v) \in D^*_w} \mX_c(-u,v) \geq c \} \sim \psi(c) \Delta_c^{-2d} \int_{D^*_w}
H(-u,v) d(-u,v)
\end{equation}
uniformly over $I(w,|\log t_n|) \subseteq t_n^{-1} \cX$ and $P\in\mathcal{P}$, where $H$ is defined in that paper and, as shown below, takes the form
\begin{equation} \label{Huv}
H(-u,v) = 4^{-2d} \text{vol}(v-u)^{-2}
\end{equation}
in our case.
Conditions (C) and (A1)--(A2) of \citet{chan_maxima_2006} are verified in Section \ref{c_a1_a2_sec}, and conditions (A3)--(A5) are verified in Section \ref{a3_a5_sec}.

We partition $t_n^{-1} \cX$ into cubes of length $|\log t_n|$ and apply (\ref{Psup}) on each
cube to show (\ref{lP}). More specifically, define $Q_n=\{ w \in (|\log t_n| \mathbb{Z})^d:
I(w,|\log t_n|) \subseteq t_n^{-1} \cX \}$. Since $\cX$ is Jordan measurable and $t_n | \log t_n|
\rightarrow 0$, 
\begin{equation} \label{hexQ}
\# Q_n \sim \text{vol}(\cX)/(t_n | \log t_n |)^d, 
\end{equation}
and it follows from (\ref{c2}), (\ref{Psup}) and (\ref{Huv}) that
\begin{equation} \label{qw}
\sum_{w \in Q_n} q_{w,P} \rightarrow \lambda \equiv (1-L^{-1})^d e^{-\zeta}
\end{equation}
uniformly over $P\in\mathcal{P}$,
noting that $\lambda$ is the limit of 
$\psi(c) \Delta_c^{-2d} (\#Q_n) |\log t_n|^d \int_{[0, L)^d}
\text{vol}(t)^{-2} dt$. Since $\mX_c$ is
independent over $D_{w_1}^*$ and $D_{w_2}^*$ for $w_1, w_2 \in Q_n$, $w_1 \neq
w_2$, it follows from the Poisson limit of the Binomial distribution that
$$P \{ \sup_{w \in Q_n} \sup_{(-u,v) \in D_w^*} \mX_c(-u,v) \geq c \} \rightarrow 1-e^{-\lambda}
$$
uniformly over $P\in\mathcal{P}$.
Hence to show (\ref{lP}), it suffices for us to prove the following:

\begin{lemma} \label{lem1}
{\rm (a)}
For all $\epsilon > 0$, there exists $L$ large enough such that
$$p_1\equiv \sup_{P\in\mathcal{P}} P \{ \sup_{(-u,v) \in D_c \setminus D_{c,L}} \mX_c(-u,v) \geq c \} \leq
\epsilon \mbox{ for all large } c.
$$

{\rm (b)}
$p_2 \equiv \sup_{P\in\mathcal{P}} \sum_{w_1,w_2 \in Q_n, w_1 \neq w_2} P \{ \sup_{u \in I(w_1,|\log t_n|), v \in 
I(w_2,|\log t_n|), {\bf 1} \leq v-u \leq L {\bf 1}} \mX_c(-u,v) \geq c \} \rightarrow 0$.

\smallskip
{\rm (c)}
 $p_3 \equiv \sup_{P\in\mathcal{P}} P \{ \sup_{(-u,v) \in D_L \setminus \cup_{w \in Q_n} 
I(w,|\log t_n|)} \mX_c(-u,v) \geq c
\} \rightarrow 0$.
\end{lemma}

We prove this lemma in Section \ref{lem1_sec}.  Sections \ref{c_a1_a2_sec} and \ref{a3_a5_sec} verify the conditions of \citet{chan_maxima_2006} for the tail approximations used in the above argument.  Section \ref{multi_y_sec} extends the results to multivariate $Y_i$.

\subsection{On (\ref{Huv}) and the Verification of (C), (A1) and (A2)}
  \label{c_a1_a2_sec}

Let $\Phi$ be the c.d.f. of the standard normal.

\begin{lemma} \label{lem2} 

{\rm (a)} Let $S_n = U_1 + \cdots + U_n$ and $s_n^2 = {\rm Var}(S_n)$. Assume that $U_1, \ldots,
U_n$ are independent mean 0 random variables and there exists $\lambda > 0$, $M_\lambda < \infty$
and $\sigma_0^2 > 0$ such that
$$E(e^{\lambda |U_k|}) \leq M_\lambda, \quad {\rm Var}(U_k) \geq \sigma_0^2, \quad 1 \leq k \leq n.
$$
Let $1 \leq x_n =o(n^{1/6})$. Then there exists a constant $C>0$ dependent only on $\lambda$,
$M_\lambda$, $\sigma_0$ and $\{ x_n \}_{n \ge 1}$ such that
$$\Big| \frac{P(S_n > x s_n)}{1-\Phi(x)} - 1 \Big| \leq \frac{C x^3}{\sqrt{n}} \mbox{ for all }
1 \leq x \leq x_n, n \geq 1. 
$$

{\rm (b) [(A1) of \citet{chan_maxima_2006}]} 
$P \{ \mathbb{H}_n(s,t) \geq c-y/c \} \sim \psi(c-y/c) [\sim 1-\Phi(c-y/c)]$ 
uniformly over $P\in\mathcal{P}$ and positive, bounded values of $y$ and $(s,t) \in B_n$. 
\end{lemma}
\begin{proof}
The special case of i.i.d. $U_k$ in (a) reduces to Theorem 1 in Chapter 16.6 of 
\citet{feller_introduction_1971}. Theorem 3 in Chapter 16.7 of \citet{feller_introduction_1971}
extends Theorem 1 to non-identically distributed random variables $U_k$ such that
$E(|U_k|^3)/E(U_k^2)$ are uniformly bounded, with a $O(\frac{x^3}{s_n})$ instead of a 
$\frac{Cx^3}{\sqrt{n}}$ error bound. We follow step-by-step the proof of Feller's Theorem 3, 
using the additional condition ${\rm Var}(U_k) \geq \sigma_0^2$ to obtain the 
$\frac{Cx^3}{\sqrt{n}}$ error bound in (a).

Under Assumption \ref{assump_a1}, $\mathbb{H}_n(s,t)=\frac{S^*}{\sigma_n(s,t)}$, where $S^*$ is
a sum of independent mean 0 random variables satisfying (i) and (ii)
with the bounds uniform over $P\in\mathcal{P}$
and ${\rm Var}(S^*) = 
\sigma_n^2(s,t)$. Hence by (a),
 \begin{equation} \label{3.2}
\frac{P \{ \mathbb{H}_n(s,t) \geq c-y/c \}}{1-\Phi(c-y/c)} = 1+O \Big( \frac{(c-y/c)^3}
{\sigma_n(s,t)} \Big) \mbox{ as } c-y/c \rightarrow \infty  %
\end{equation}
uniformly over $P\in\mathcal{P}$ and $(s,t) \in B_n$.
Since vol($I(s,t)) \geq t_n^d$ for $(s,t) \in B_n$,
by Assumption \ref{assump_a2}(b), 
$$\liminf_{n \rightarrow \infty} [\inf_{(s,t) \in B_n} 
\# J_n(s,t)]/(nt_n^d) \geq \inf_{x \in \cal X} f(x) > 0.
$$
Hence by Assumption \ref{assump_a1}(b) and \ref{assump_a2}(a), $\sigma_n^{-2}(s,t) = O((nt_n^d)^{-1}) = 
o(| \log t_n|^{-4})$ uniformly over $(s,t) \in B_n$ and $P\in\mathcal{P}$. Since
$c = O(| \log t_n|^{1/2})$, (b) follows from (\ref{3.2}). %
\end{proof}

\medskip

Let $\rho_c(-u,v,-u_1,v_1) = {\rm Cov}(\mathbb{X}_c(-u,v),\mathbb{X}_c(-u_1,v_1))$
(we suppress the dependence of $\rho_c$ on $P$ in the notation)
and let                
$\{ W_{-u,v}(q,r): (q,r) \in [0,\infty)^{2d} \}$ be a 
continuous Gaussian random field satisfying
\begin{eqnarray} \label{Wuv}
& & W_{-u,v}(0)=0, \quad E[W_{-u,v}(q,r)] = -\sum_{j=1}^d \frac{q_j+r_j}
{4(v_j-u_j)}, \\
& & {\rm Cov}(W_{-u,v}(q,r), W_{-u,v}(\alpha,\beta)) = \sum_{j=1}^d
\frac{\min(q_j,\alpha_j)+\min(r_j,\beta_j)}{2(v_j-u_j)}. \nonumber
\end{eqnarray}

\begin{lemma} \label{lem3}

{\rm (a)} {\rm [(C) of \citet{chan_maxima_2006}]} 
$1-\rho_c(-u,v,-u+\delta_u,v+\delta_v) \sim \sum_{j=1}^d \frac{\delta_{u,j}+\delta_{v,j}}
{2(v_j-u_j)}$ uniformly over $(-u,v) \in D_L$
and $P\in\mathcal{P}$
and compact sets of $(\delta_u,\delta_v)/
\Delta_c > 0$. 

\smallskip
{\rm (b)} {\rm [(A2) of \citet{chan_maxima_2006}]} 
For any $a>0$ and positive integer $m$, as $c \rightarrow \infty$,
\begin{eqnarray*}
& & \{ c[\mX_c(-u+ak_u \Delta_c,v+ak_v \Delta_c)-\mX_c(-u,v)]: 0\leq (k_u,k_v) < m {\bf 1} \}
| \mX_c(-u,v) = c-y/c \cr
& \stackrel{d}{\to}& \{ W_{-u,v}(ak_u,ak_v): 0\leq
(k_u, k_v) < m {\bf 1} \},
\end{eqnarray*}
uniformly over $(-u,v) \in D_L$
and $P\in\mathcal{P}$
and positive bounded values of $y$.

\smallskip
{\rm (c)} $H(-u,v)\equiv\lim_{K \rightarrow \infty} \int_0^\infty e^y P \{ \sup_{{\bf 0} \leq (q,r) \leq K {\bf 1}}
W_{-u,v}(q,r) \geq y \} dy$ has the closed-form given in {\rm (\ref{Huv})}.
\end{lemma}
\begin{proof}
Let $z_0 = (s,t)$, where $s=ut_n$, $t=(v-u)t_n$ and $z_{\delta} = (s-\delta_u t_n,
t+(\delta_v+\delta_u)t_n)$ for some $\delta_u, \delta_v \geq 0$. Then 
\begin{eqnarray} \label{rhoc}
& & \rho_c(-u,v,-u+\delta_u,v+\delta_v) = \sigma_n(z_0)/\sigma_n(z_\delta) \\ 
& = & \Big(1+\frac{\sigma_n^2(z_\delta)-\sigma_n^2(z_0)}{\sigma_n^2(z_0)} \Big)^{-1/2} = 
1-\frac{\sigma_n^2(z_\delta)-\sigma_n^2(z_0)}{2 \sigma_n^2(z_0)} + O \Big(
\frac{\sigma_n^2(z_\delta)-\sigma_n^2(z_0)}{\sigma_n^2(z_0)} \Big)^2. \nonumber
\end{eqnarray}
Since $\Delta_c \sim (4 | \log t_n |)^{-1}$, by Assumption \ref{assump_a2}, 
\begin{equation} \label{sig0}
\sigma_n^2(z_0) \sim n \sigma^2(s) f(s) \text{vol}(t), \quad [\sigma_n^2(z_\delta)-\sigma_n^2(z_0)] 
\sim n \sigma^2(s) f(s) \text{vol}(t) \sum_{j=1}^d \frac{\delta_{u,j}+\delta_{v,j}}
{v_j-u_j},
\end{equation}
and (a) follows from substituting (\ref{sig0}) into (\ref{rhoc}). 

Let $a>0$ and let $\delta_u=ak_u \Delta_c$, $\delta_v=ak_v \Delta_c$. Then
\begin{equation} \label{cXc}
c[\mX_c(-u+\delta_u,v+\delta_v)-\mX_c(-u,v)] = c \Big\{ \frac{\sum_{i \in J_n(z_\delta) \setminus
J_n(z_0)} Z_i}{\sigma_n(z_\delta)} + \mX_c(-u,v) \Big[ \frac{\sigma_n(z_0)}{\sigma_n(z_\delta)}
-1 \Big] \Big\}.
\end{equation}
We note here that as $\delta_u, \delta_v \geq 0$, so $J_n(z_\delta) \supseteq J_n(z_0)$.
By (\ref{Wuv})--(\ref{sig0}), conditioned on $\mX_c(-u,v)=c-y/c$ and noting that 
$c^2 \Delta_c = \frac{1}{2}$, 
\begin{eqnarray} \label{cX}
c \mX_c(-u,v) \Big[ \frac{\sigma_n(z_0)}{\sigma_n(z_\delta)}-1 \Big] & \rightarrow & 
\sum_{j=1}^d \frac{\delta_{u,j}+\delta_{v,j}}{4(v_j-u_j)} = E[W_{-u,v}(ak_u,ak_v)], \\
\label{VX}
{\rm Var} \Big( \frac{c \sum_{i \in J_n(z_\delta) \setminus J_n(z_0)} Z_i}{\sigma_n(z_\delta)}
\Big) & = & \frac{c^2[\sigma_n^2(z_\delta)-\sigma_n^2(z_0)]}{\sigma_n^2(z_\delta)}
\rightarrow {\rm Var}(W_{-u,v}(ak_u,ak_v)).
\end{eqnarray}
Similarly, if $z_{\tilde \delta} = (s-\widetilde \delta_u t_n,t+(\widetilde \delta_u +
\widetilde \delta_v) t_n)$, where $\delta_u = a \widetilde k_u \Delta_c$, $\delta_v = a
\widetilde k_v \Delta_c$ with $\widetilde k_u, \widetilde k_v \geq 0$, then
\begin{eqnarray} \label{rCov}
{\rm Cov} \Big( \frac{c \sum_{i \in J_n(z_\delta) \setminus J_n(z_0)}
Z_i}{\sigma_n(z_\delta)}, \frac{c \sum_{i \in J_n(z_{\tilde \delta}) \setminus J_n(z_0)}
Z_i}{\sigma_n(z_{\tilde \delta})} \Big) & = & 
\frac{c^2[\sigma_n^2(z_{\min(\delta, \tilde \delta)})
-\sigma_n^2(z_0)]}{\sigma_n(z_\delta) \sigma_n(z_{\tilde \delta})} \\
& \rightarrow & {\rm Cov}(W_{-u,v}(ak_u,ak_v), W_{-u,v}(a \widetilde k_u, a \widetilde k_v)).
\nonumber
\end{eqnarray}
Since $\sum_{i \in J_n(z_\delta) \setminus J_n(z_0)} Z_i$ is independent of $\mX_c(-u,v)$ and
is asymptotically normal by Assumptions \ref{assump_a1}(b)--(c) and \ref{assump_a2}(b), (b) follows from (\ref{cX})--(\ref{rCov}).
Lastly, 
(c) is a direct consequence of Lemma 2.3 of \citet{chan_maxima_2006}. %
\end{proof}

\subsection{Proof of Lemma \ref{lem1}}\label{lem1_sec}

To deal with technicalities associated with non-rectagular edges, we 
extend the domain of $\mathbb{H}_n$ to ${\cal C} \times [t_n,1)^d$ for some ${\cal C} = 
[-C,C]^d$ by embedding
$(x_1,Y_1), (x_2,Y_2), \ldots$ as a subsequence of $(\widetilde x_1, \widetilde Y_1), (\widetilde 
x_2, \widetilde Y_2), \ldots$ with $\widetilde x_i \in [-(C+1),(C+1)]^d$.  
Hence the domain of $\mathbb{X}_c$ can be extended to $\{ (-u,v) \in t_n^{-1} {\cal C}^2: 
{\bf 1} \leq v-u \leq t_n^{-1} {\bf 1} \}$ with (C) and (A1)--(A5) 
satisfied uniformly over $\{ (-u,v) \in t_n^{-1}
{\cal C}^2: {\bf 1} \leq v-u \leq L {\bf 1} \}$ for any fixed $L > 1$. 

\begin{proof}[Proof of Lemma \ref{lem1}(c)] 
Let $\widetilde 
Q_n = \{ w \in (|\log t_n| \mathbb{Z})^d: I(w,|\log
t_n|) \cap (t_n^{-1} \cX) \neq \emptyset \}$. Since $\cX$ is Jordan measurable, $\# Q_n \sim
\# \widetilde Q_n$. Hence by (\ref{Psup}) and (\ref{qw}), 
$$p_3 \leq \sup_{P \in \mathcal{P}} \sum_{w \in \tilde Q_n \setminus Q_n} q_{w,P} 
= o \Big( \sup_{P \in \mathcal{P}} \sum_{w \in Q_n} q_{w,P} \Big) = o(1). %
$$
\end{proof}

\begin{proof}[Proof of Lemma \ref{lem1}(b)]
For $n$ large enough such that $|\log t_n|>L$,
$u \in I(w_1,|\log t_n|)$, $v \in I(w_2,|\log t_n|)$, $v-u \leq L {\bf 1}$ can occur only when
$w_1$, $w_2$ are neighboring cubes. Note that each cube has not more than $3^d-1$ neighbors. When 
$w_1$ and $w_2$ neighboring cubes, define
$$D_{w_1,w_2}^* = \{ (-u,v): u \in I(w_1,|\log t_n|), v \in I(w_2,|\log t_n|), {\bf 1} \leq v-u \leq
L {\bf 1} \}.
$$
Since vol($D_{w_1,w_2}^*) = o(| \log t_n |^{-d})$ and $H(-u,v) \leq 1$, by Corollary 2.7 of Chan and
Lai (2006), 
\begin{eqnarray*}
P \{ \sup_{(-u,v) \in D_{w_1,w_2}^*} \mX_c(-u,v) \geq c \} & \sim & \psi(c) \Delta_c^{-2d}
\int_{D_{w_1,w_2}^*} H(-u,v) d(-u,v) \cr
& = & o(\psi(c) \Delta_c^{-2d} | \log t_n |^{-d})
\end{eqnarray*}
uniformly over $P\in\mathcal{P}$ and over neighboring $w_1$ and $w_2$.
Hence by (\ref{c2}) and (\ref{hexQ}), $p_2 = o((\# Q_n) \psi(c) \Delta_c^{-2d} | \log t_n |^{-d}) =
o(1)$. %
\end{proof}

\begin{proof}[Proof of Lemma \ref{lem1}(a)]
For each $\ell \in \mathbb{Z}^d$, $\ell \neq 0$ with
$0 \leq \ell \leq [\log_L (2t_n^{-1} C)] {\bf 1}$, define
$$\mX_{c,\ell}(-u,v) = \mathbb{H}_n(ut_n L^\ell,(v-u)t_n L^\ell) \mbox{ for }
u,t \in (t_n L^\ell)^{-1} {\cal C} \mbox{ with } (v-u) \geq {\bf 1}.
$$
We use here the convention $a {\cal C} = \prod_{j=1}^d [-a_j C, a_j C]$. To avoid double
counting, we restrict the domain of $\mX_{c,\ell}$ to
$$D_{\ell} \equiv  \{ (-u,v) \in (t_n L^\ell)^{-1} {\cal C}^2: {\bf 1} \leq v-u \leq L {\bf 1} \}.
$$
By Corollary 2.7 of \citet{chan_maxima_2006}, 
\begin{equation} \label{Dl}
P \{ \sup_{(-u,v) \in D_\ell} \mX_c(-u,v) \geq c \} \sim \psi(c) \Delta_c^{-2d} \int_{D_\ell}
H(-u,v) d(-u,v) \mbox{ uniformly over } \ell \mbox{ and } P\in\mathcal{P}, 
\end{equation}
with $H(-u,v) = O(1)$ uniformly over $\ell$ and $D_\ell$. By (\ref{c2}), 
$\psi(c) \Delta_c^{-2d} \text{vol}(t_n^{-1} {\cal C}) = O(1)$ and so
\begin{equation} \label{Dl2}
\psi(c) \Delta_c^{-2d} 
\int_{D_\ell} H(-u,v) d(-u,v) = O(|L^\ell|^{-1}) \mbox{ uniformly over } \ell.
\end{equation}
Hence by (\ref{Dl}) and (\ref{Dl2}), 
$$p_1 \leq \sup_{P\in\mathcal{P}}
\sum_{0 \leq \ell \leq [\log_L (2t_n^{-1} C)] {\bf 1}, \ell \neq 0}
P \{ \sup_{(-u,v) \in D_\ell} \mX_c(-u,v) \geq c \} = O \Big( 
 \sum_{0 \leq \ell \leq [\log_L (2t_n^{-1} C)] {\bf 1}, \ell \neq 0} |L^\ell|^{-1} \Big).
$$
The sum above within $O(\cdot)$ is bounded by $(\sum_{k=0}^\infty L^{-k})^d-1 = 
(1-L^{-1})^{-d}-1$ which can be made arbitrarily small by choosing $L$ large enough. %
\end{proof}

\subsection{Extension to Multivariate $Y$}\label{multi_y_sec}

Let $E_{w,j} = \{ \sup_{(-u,v) \in D_w^*} \mX_{c,j} \geq c \}$, where $c = \frac{b_n+
\min_{1 \leq j \leq d_Y} \zeta_j \sqrt{n}}{a_n}$, for given $\zeta_1, \ldots,
\zeta_{d_Y}$, see (\ref{Dwstar}). To extend
the proof of Theorem \ref{asym_dist_thm} to $d_Y>1$, it suffices to prove the following:

\begin{lemma} \label{lem4}
$p_4 \equiv \sup_{P \in\mathcal{P}} \sum_{w \in Q_n} \sum_{j_1 \neq j_2} P(E_{w,j_1} \cap E_{w,j_2}) \rightarrow 0$.
\end{lemma}
\begin{proof}
Fix $w$ and partition $D_w^*$ into cubes of length $\Delta_c$. More specifically, define
$K_c = \{ z \in (\Delta_c \mathbb{Z})^{2d}: I(z,\Delta_c {\bf 1}) \cap D_w^* \} \neq \emptyset$.
Let $G_{z,j} = \{ \sup_{(-u,v) \in I(z,\Delta_c {\bf 1})} \mX_{c,j}(-u,v) \geq c \}$. Then,
uniformly over $z$, $P\in\mathcal{P}$ and $1 \leq j \leq d_Y$,
\begin{eqnarray} \label{PG}
& & P(G_{z,j} \cap \{ \mX_{c,j}(z) \leq c-\theta/c \}) \sim \psi(c) H_{\theta}(z), \\
& \mbox{where} & H_{\theta}(z) = \int_\theta^\infty e^y P \{ \sup_{0 \leq w \leq {\bf 1}}
W_z(w) > y \} dy. \nonumber
\end{eqnarray}
This extends Theorem 2.4 of \citet{chan_maxima_2006} to $\theta \neq 0$, using the
same proof. Since $H_0(z) < \infty$, for any 
given $\epsilon > 0$, we can select $\theta$ large enough such that 
$H_{\theta}(z) \leq \epsilon$. In addition, by (\ref{Wuv}), 
this selection can be made to be uniform over
$z \in K_c$ and $1 \leq j \leq d_Y$. Note that
\begin{eqnarray*}
P(E_{w,j_1} \cap E_{w,j_2}) & \leq & P(\mX_{c,j_1}(z_1) > c-\theta/c, \mX_{c,j_2}(z_2)
> c-\theta/c \mbox{ for some } z_1, z_2 \in K_c) + \eta_{c,w}, \cr
\mbox{ where } \eta_{c,w} & = & \sum_{z \in K_c} [P(G_{z,j_1} \cap \{ \mX_{c,j_1}(z) \leq
c-\theta/c \}) + P(G_{z,j_2} \cap \{ \mX_{c,j_2}(z) \leq c-\theta/c \})],
\end{eqnarray*}
and with $\theta$ selected so that $H_{\theta}(z) \leq \epsilon$, it
follows from (\ref{PG}) that
$$\eta_{c,w} = \epsilon O(\psi(c)(\# K_c)) = \epsilon O(\psi(c) \Delta_c^{-2d} | \log t_n |^d).
$$
By (\ref{c2}), $\psi(c) = O(t_n^d \Delta_c^{2d})$ and hence by 
(\ref{hexQ}), $\sum_{w \in Q_n} \eta_{c,w} = \epsilon O(1)$. It remains for
us to show that for all $\theta > 0$,
\begin{equation} \label{Szz}
\sum_{z_1, z_2 \in K_n} P(\mX_{c,j_1}(z_1) > c-\theta/c, \mX_{c,j_2}(z_2) > c-\theta/c)
=o(\psi(c) \Delta_c^{-2d} | \log t_n |^d).
\end{equation}
Now by Assumption \ref{assump_a1}(d), $S(z_1,z_2)\equiv\mX_{c,j_1}(z_1) \cap \mX_{c,j_2}(z_2)$ has mean 0 and
variance lying between $2(1-\rho)$ and $2(1+\rho)$. Let $\kappa = (\frac{2}{1+\rho})^{1/2} (>1)$. By Lemma \ref{lem2}(a),
\begin{equation} \label{PS}
P \{ S(z_1,z_2) \geq 2(c-\theta/c) \} \leq \Big[1+O \Big( \frac{c^3}{\sqrt{n t_n^d}} 
\Big) \Big] [1-\Phi(\kappa(c-\theta/c))] \sim
\frac{1}{\kappa c (2 \pi)^{1/2}} e^{-\kappa^2 c^2/2+\kappa \theta}
\end{equation}
uniformly over $P\in\mathcal{P}$.
Since $\# K_c =O(\Delta_c^{-2d} | \log t_n |^d)$, it follows from (\ref{PS}) that
\begin{eqnarray*}
& & \sum_{z_1,z_2 \in K_c} P(\mX_{c,j_1}(z_1) > c-\theta/c, \mX_{c,j_2}(z_2) > c-\theta/c) \cr
& \leq & \sum_{z_1,z_2 \in K_c} P \{ S(z_1,z_2) > 2(c-\theta/c) \} = O(\psi(c) e^{-(\kappa^2-1)c^2/2}
\Delta_c^{-4d} | \log t_n |^{2d})
\end{eqnarray*}
uniformly over $P\in\mathcal{P}$
and (\ref{Szz}) holds because $| \log t_n |^d = O(c^{2d})$ and $c^{6d} e^{-(\kappa^2-1)c^2/2}
=o(1)$. %
\end{proof}

\subsection{Verification of (A3)--(A5)}\label{a3_a5_sec}

Conditions
(C), (A1) and (A2) have been verified in Section \ref{c_a1_a2_sec}. The remaining 
regularity conditions that lead to (\ref{Psup}) will be verified in Lemmas \ref{lem5} and 
\ref{lem6} below.

\begin{lemma} \label{lem5}

{\rm (a) [(A3) of \citet{chan_maxima_2006}]} Let $\gamma > 0$ and $k_u, k_v \geq 0$. 
There exists a positive function $h$ such that $\lim_{y \rightarrow \infty} h(y) = 0$ and
$$P \{ \mX_c(-u+k_u \Delta_c,v+k_v \Delta_c) > c-\gamma/c, \mX_c(-u,v) \leq c-y/c \}
\leq h(y) \psi(c) \mbox{ for all large } c,
$$
uniformly over $(-u,v) \in D_L$ and $P\in\mathcal{P}$.

\smallskip
{\rm (b) [(A5) of \citet{chan_maxima_2006}]} There exists a nonincreasing positive function $r$ on
$[0,\infty)$ such that $r(\| k \|) =O(e^{-\|k \|^p})$ for some $p>0$ such that for any $\gamma > 0$,
$$P \{ \mX_c(-u,v) > c-\gamma/c, \mX_c(-u+k_u \Delta_c, v+k_v \Delta_c) > c- \gamma/c \}
\leq \psi(c-\gamma/c) r(\| k_u,k_v \|) \mbox{ for all large } c,
$$
uniformly over $P\in\mathcal{P}$, $(-u,v), (-u+k_u \Delta_c,v+k_v \Delta_c) \in D_w^*$ and $w \in Q_n$.
\end{lemma}
\begin{proof}
Let $\omega > 1$ to be specified later. By Lemma \ref{lem2}(a),
there exists $\xi_c \rightarrow 0$ such that 
$$P \{ \mX_c(-u,v) \geq c-y'/c \} = [1+O(\xi_c^2)] e^{y'} \psi(c)
$$
uniformly over $\gamma \leq y' \leq \omega c$ and $P \in \mathcal{P}$. 
Let $y_j=y+j \xi_c$, $j=0,1,\ldots$. 
Let $u_1 = u-k_u \Delta_c$ and $v_1 = v+k_v
\Delta_c$. Since $e^{\xi_c} = 1+\xi_c+O(\xi_c^2)$, 
\begin{eqnarray} \label{A3.4}
& & P \{ \mX_c(-u,v) > c-y_{j+1}/c \} - P \{ \mX_c(-u,v) > c-y_j/c \} \\
& = & [1+O(\xi_c^2)]e^{y_j+\xi_c} \psi(c) - [1+O(\xi_c^2)]e^{y_j} \psi(c) \sim 
\xi_c e^{y_j} \psi(c) \nonumber
\end{eqnarray}
uniformly over $\gamma \leq y_j \leq \omega c$ and $P \in \mathcal{P}$. 
Since $P \{ \mX_c(-u_1,v_1)>a|\mX_c(-u,v)=b \}$ increases with $b$ for
any fixed $a$, it follows from (\ref{A3.4}) that
\begin{eqnarray} \label{A3.2}
& & P \{ \mX_c(-u_1,v_1) > c-y/c, c-\omega \leq \mX_c(-u,v) < c-y/c \} \\
& \leq & \sum_{0 \leq j \leq (\omega c-y)/\xi_c} P \{ \mX_c(-u_1,v_1) > c-\gamma/c|\mX_c(-u,v)=
c-y_j/c \} \nonumber \\
& & \times [P \{ \mX_c(-u,v) > c-y_{j+1}/c \}-P \{ \mX_c(-u,v) > c-y_j/c \}] \nonumber \\
& \sim & \psi(c) \xi_c \sum_{0 \leq y_j \leq \omega c} e^{y_j} P \{ \mX_c(-u_1,v_1) > c-\gamma/c|\mX_c(-u,v)=
c-y_j/c \}. \nonumber
\end{eqnarray}

Let $z_\delta = (u_1 t_n, (v_1-u_1)t_n)$, $k=(k_u,k_v)$ and let 
\begin{equation} \label{gk}
g_{-u,v}(k) = \sum_{j=1}^d \frac{k_{u,j}+k_{v,j}}{4(v_j-u_j)} \{ = E[W_{-u,v}(k_u,
k_v)] = {\rm Var}(W_{-u,v}(k_u,k_v))/2 \},
\end{equation}
(see (\ref{Wuv})). Then by (\ref{sig0})--(\ref{VX}) with $a=1$, 
\begin{eqnarray} \label{a1}
& & P \{ c[\mX_c(-u_1,v_1)-\mX_c(-u,v)] \geq y_j-\gamma|\mX_c(-u,v) = c-y_j/c \} \\
& = & P \Big\{ \frac{\sum_{i \in J_n(z_\delta) \setminus J_n(z_0)} Z_i}{\sqrt{ \sigma_n^2
(z_\delta)-\sigma_n^2(z_0)}} \geq \frac{y_j-\gamma-c(c-y_j/c)
[\frac{\sigma_n(z_0)}{\sigma_n(z_\delta)}
-1]\sigma_n(z_\delta)}{c \sqrt{\sigma_n^2(z_\delta)-\sigma_n^2(z_0)}} \Big\} \nonumber \\
& = & P \Big\{ \frac{\sum_{i \in J_n(z_\delta) \setminus J_n(z_0)} Z_i}{\sqrt{ \sigma_n^2
(z_\delta)-\sigma_n^2(z_0)}} \geq \frac{y_j
-\gamma+g_{-u,v}(k)+o(1)}{\sqrt{2g_{-u,v}(k)+o(1)}} \Big\} \sim 
\psi \Big( \frac{y_j-\gamma+g_{-u,v}(k)}
{\sqrt{2 g_{-u,v}(k)}} \Big), \nonumber
\end{eqnarray}
with $o(1)$ uniform over $y \leq y_j \leq \omega c$ and $(-u,v) \in D_L$ and $P\in\mathcal{P}$, noting that as 
$y'=O(c)$, the relative error of 
the normal tail approximation in (\ref{a1}) is
$$O \Big( \frac{c^3}{\sqrt{\sigma_n^2(z_\delta)-\sigma_n^2(z_0)}} \Big) = O \Big( \frac{c^4}{\sqrt{n
t_n^d}} \Big) \rightarrow 0
$$
(see Assumption \ref{assump_a2}(a) and Lemma \ref{lem2}(a)). By (\ref{A3.2}) and (\ref{a1}), 
\begin{eqnarray} \label{md}
&  & P \{ \mX_c(-u_1,v_1) > c-\gamma/c, c-\omega < \mX_c(-u,v) \leq c-y/c \} \\
& \sim & 
\psi(c) \int_y^{\omega c} e^{y'} \psi \Big( \frac{y'-\gamma+g_{-u,v}(k)}{\sqrt{2g_{-u,v}(k)}}
\Big) dy'. \nonumber
\end{eqnarray}

To complete the proof of (a), it suffices to show that
\begin{equation} \label{(II)}
{\rm (II)}\equiv P \{ \mX_c(-u_1,v_1) > c-\gamma/c, \mX_c(-u,v) \leq c-\omega \} = o(\psi(c))
\end{equation}
for all $\omega$ large. By (\ref{Hn}) and (\ref{Xc}), 
\begin{equation} \label{Zisum}
\sum_{i \in J_n(z_\delta) \setminus J_n(z_0)} Z_i = \sigma_n(z_\delta)
\mX_c(-u_1,v_1)-\sigma_n(z_0) \mX_c(-u,v).
\end{equation}
Since $\sigma_n(z_\delta) \geq \sigma_n(z_0)$, by (\ref{VX}), 
(\ref{gk}) and Lemma \ref{lem2}(a) with $a=1$, 
\begin{eqnarray*}
{\rm (II)} & \leq & P \Big\{ \frac{\sum_{i \in J_n(z_\delta) \setminus J_n(z_0)} Z_i}{\sqrt{
\sigma_n^2(z_\delta)-\sigma_n^2(z_0)}} \geq \frac{\sigma_n(z_0)(\omega c-\gamma)}
{c \sqrt{\sigma_n^2(z_\delta)-\sigma_n^2(z_0)}} \Big\} = 
\Big[1+O \Big( \frac{x_n^3}{\sqrt{nt_n^d}} \Big) \Big] \psi(x_n), \cr
\mbox{ where } x_n & = & 
\frac{\omega c-\gamma}{c \sqrt{4 \Delta_c[g_{-u,v}(k)+o(1)]}} = \frac{\omega c-\gamma}
{\sqrt{2[g_{-u,v}(k)+o(1)]}},
\end{eqnarray*}
and indeed (\ref{(II)}) holds when $\omega > \sqrt{2 g_{-u,v}(k)}$.

To prove (b), we apply Lemma \ref{lem2}(a) to the right-hand side of the inequality
$$P \{ \mX_c(-u,v) > c-\gamma/c, \mX_c(-u_1,v_1) > c-\gamma/c \} \leq 
P \{ \mX_c(-u,v)+\mX_c(-u_1,v_1) > 2(c-\gamma/c) \} {\rm [\equiv\rm (III)]}. 
$$
As in the proof of Lemma \ref{lem1}(b), the relative error of the normal approximation goes to 0
due to Assumption \ref{assump_a2}(a), that is,
\begin{equation} \label{III}
{\rm (III)} \sim \psi \Big( \frac{2(c-\gamma/c)}{\sqrt{2+2 \rho_c(-u,v,-u_1,v_1)}} \Big)
\mbox{ as } c \rightarrow \infty.
\end{equation}
Note that 
in the statement of (b), the restriction $k_u, k_v \geq 0$ is removed and we have in place of
(\ref{rhoc}), 
$$\rho_c(-u,v,-u_1,v_1) = \frac{\sigma_n^2(z^*)}{\sigma_n(z_0) \sigma_n(z_\delta)},
$$
where $z^* = (-(u\vee u_1),(v \wedge v_1))$.
Since $J_n(z^*) = J_n(z_0) \cap J_n(z_\delta)$, so by expanding $\sigma_n(z^*)/\sigma_n(z_0)$
and $\sigma_n(z^*)/\sigma_n(z_\delta)$ as in (\ref{rhoc}), it follows from (\ref{sig0}) and
(\ref{gk}) with
$\delta_u = k_u \Delta_c$, $\delta_v = k_v \Delta_c$ that
\begin{eqnarray*}
\rho_c(-u,v,-u_1,v_1) & = & 1-(1+o(1)) \Big\{ \frac{\sigma_n^2(z_\delta)-\sigma_n^2(z^*)}
{2 \sigma_n^2(z^*)}+\frac{\sigma_n^2
(z_0)-\sigma_n^2(z^*)}{2 \sigma_n^2(z^*)} \Big\} \cr
& = & 1-(1+o(1)) \Big\{ \sum_{j=1}^d \frac{(\delta_{u,j})^+ + (\delta_{v,j})^+}{v_j-u_j}
+ \sum_{j=1}^d \frac{(\delta_{u,j})^- + (\delta_{v,j})^-}{v_j-u_j} \Big\} \cr
& = & 1-(4+o(1)) \Delta_c g_{-u,v}(|k|),
\end{eqnarray*}
from which it follows that 
$$\frac{2(c-\gamma/c)}{\sqrt{2+2 \rho_c(-u,v,-u_1,v_1)}} = (c-\gamma/c)[1-(2+o(1)) \Delta_c 
g_{-u,v}(|k|)]^{-1/2} \geq c+[g_{-u,v}(|k|)/2-\gamma+o(1)]/c,
$$
and (b) with $r(\tau) = \exp[-\min_{\| k \|=\tau} g_{-u,v}(|k|)/4]$
and $0 < p < 1$, follows from (\ref{III}). %
\end{proof}

\begin{lemma} \label{lem6}

{\rm (a) [Theorem 1 of \citet{wichura_inequalities_1969}]} Let $\cal A$ be a finite subset of $\mathbb{R}^d$ and let $U_i$,
$i \in {\cal A}$ be independent mean {\rm 0} random variables with variance $\sigma_i^2$. Let
$S_k = \sum_{i \leq k} U_i$ and set $s_{\cal A}^2 = 
\sum_{i \in \cal A} \sigma_i^2$, $S_{\cal A} = \sum_{i \in \cal A} U_i$. Then for any $x > 2^d
s_{\cal A}$,
\begin{equation} \label{PS}
P(\max_{k \in \mathbb{R}^d} |S_k| > x) \leq [1-(2^d s_{\cal A}/x)^2]^{-d} P(|S_{\cal A}| > 2^{-d} x).
\end{equation}

{\rm (b) [(A4) of \citet{chan_maxima_2006}]} There exists nonincreasing functions $N_a$ on 
$\mathbb{R}^+$ and positive constants $\gamma_a \rightarrow 0$ such that $N_a(\gamma_a) + 
\int_1^\infty \tau^s N_a(\gamma_a+\tau) d \tau = o(a^d)$ as $a \rightarrow 0$ for all
$s>0$, and for each $a > 0$,
\begin{equation} \label{l6}
P \{ \sup_{0 \leq (k_u,k_v) \leq a {\bf 1}} \mX_c(-u+k_u \Delta_c,v+k_v \Delta_c) > c,
\mX_c(-u,v) \leq c-\gamma/c \} \leq N_a(\gamma) \psi(c),
\end{equation}
uniformly over $(-u,v) \in D_L$ and $P\in\mathcal{P}$
for all $\gamma_a \leq \gamma \leq c$ with $c$ large.
\end{lemma}
\begin{proof}
Though \citet{wichura_inequalities_1969} considers a set $\cal A$ with points lying on a $d$-dimensional
grid, we can always extend $\cal A$ to a $d$-dimensional grid $\cal B$ by letting $U_i \equiv 0$
for $i \in {\cal B} \setminus {\cal A}$. Note that the right-hand side of (\ref{PS}) is unchanged
by such an extension. Let $u_1=u-k_u \Delta_c$, 
$v_1=v+k_v \Delta_c$, $k=(k_u,k_v)$ and $z_\delta = (u_1 t_n, (v_1-u_1)t_n)$. Let
$\omega >1$ to be specified later. Since $\sigma_n(z_\delta) \geq \sigma_n(z_0)$ when
$J_n(z_\delta) \supseteq J_n(z_0)$, by the arguments in (\ref{A3.2}) and (\ref{a1}),
\begin{eqnarray} \label{6I}
{\rm (I)} & \equiv & P \{ \sup_{0 \leq k \leq a {\bf 1}} \mX_c(-u_1,v_1) > c,
c-\omega \leq \mX_c(-u,v) \leq c-\gamma/c \} \\
& \sim & \psi(c) \int_\gamma^{\omega c} e^y P \Big\{ \sup_{0 \leq k \leq a {\bf 1}} 
\sum_{i \in J_n(z_\delta) \setminus J_n(z_0)} Z_i \geq \sigma_n(z_0) y/c \Big\} dy. \nonumber
\end{eqnarray}

Let $\cal B$ be the set of all $d$-dimensional vectors with coordinates taking values
$-1$, 0 or 1 but not all zeros. Hence $\# {\cal B} = 3^d-1$. Consider the partitioning of
${\cal A}\equiv J_n((u-a \Delta_c)t_n, (v-u+2a \Delta_c)t_n) \setminus J_n(z_0)$ as ${\cal A} = 
\bigcup_{b \in \cal B}
{\cal A}_b$, with
\begin{eqnarray*}
{\cal A}_b = \{ i: (u_j-a \Delta_c)t_n & \leq & x_{i,j} < u_j t_n \mbox{ if }
b_j=-1, \cr
u_j t_n & \leq & x_{i,j} \leq v_j t_n \mbox{ if } b_j=0, \cr
v_j t_n & \leq & x_{i,j} \leq (v_j+a \Delta_c) t_n \mbox{ if } b_j=1 \}.
\end{eqnarray*}
Then
$$\sup_{0 \leq k \leq a {\bf 1}} \sum_{i \in J_n(z_\delta) \setminus J_n(z_0)}
Z_i \leq \sum_{b \in \cal B} \max_{k \in \mathbb{R}^d} S_{k,b}, \mbox{ where } S_{k,b} = 
\sum_{i \in {\cal A}_b, i \leq k} Z_i.
$$
Let $x=\frac{\sigma_n(z_0) y}{3^d c}$. By (\ref{sig0}) , and since $v_j - u_j \geq 1$,
\begin{eqnarray} \label{xs}
\frac{x}{s_{\cal A}} \sim \frac{y \text{vol}(v-u)}{3^d c(\text{vol}(v-u+2a \Delta_c {\bf 1})-\text{vol}(v-u))^{1/2}} & = & 
\frac{y}{3^d c} \Big[ \prod_{j=1}^d \Big(1+\frac{2a \Delta_c}{v_j-u_j} \Big)-1
\Big]^{-1/2} \\
& \geq & (1+o(1)) \frac{y}{3^d c} \Big( \frac{da}{c^2} \Big)^{-1/2} \sim \frac{y}{3^d \sqrt{da}}.
\nonumber
\end{eqnarray}
Hence for all large $c$, 
$$\frac{x}{s_{\cal A}} \geq 2^{d+1/2} \mbox{ when } y \geq \gamma \mbox{ for } a \leq
\Big( \frac{\gamma}{6^{d+1/2}} \Big)^2 \frac{1}{d}.
$$
By (\ref{PS}) and (\ref{6I}), and since $s_{{\cal A}_b} \leq s_{\cal A}$,
\begin{equation} \label{s1}
{\rm (I)} \leq (2^d+o(1)) \psi(c) \sum_{b \in {\cal B}} \int_\gamma^{\omega c} e^y P \Big\{
|S_{{\cal A}_b}| \geq \frac{\sigma_n(z_0) y}{6^d c} \Big\} dy.
\end{equation}
Apply Lemma \ref{lem2}(a) and note that the sum in (\ref{s1}) is dominated by the $2d$ values
of $b$ having a single non-zero entry. Then by (\ref{xs}), (\ref{l6}) holds (for large $c$ and
small $a$) when
$$N_a(\gamma) = d 2^{d+2} \int_\gamma^\infty e^y \psi \Big( \frac{y}{6^{d+1/2} \sqrt{da}} \Big) dy.
$$
We check that when $\gamma_a=a^{1/3}$, then $N_a(\gamma_a)+\int_1^\infty
\tau^s N_a(\gamma_a+\tau)=o(a^p)$ as $a \rightarrow 0$ for all $s>0$ and $p>0$, and that
$$P \{ \sup_{0 \leq k \leq a {\bf 1}} \mX_c(-u_1,v_1)>c,
\mX_c(-u,v) \leq c-\omega \} \leq P \Big\{ \sup_{0 \leq k \leq a {\bf 1}}
\sum_{i \in J_n(z_\delta) \setminus J_n(z_0)} Z_i \geq \omega \sigma_n(z_0) \Big\} = o(\psi(c))
$$
for all $\omega$ large, by a similar partitioning argument and applications of Lemmas
\ref{lem2}(a) and \ref{lem6}(a). %
\end{proof}

\section{Proof of Theorem \ref{local_alt_thm}}

We have
\begin{align*}
&\bar m_j(\theta_n,x)=\bar m_j(\theta_0,x_0)
  +\bar m_j(\theta_n,x)-\bar m_j(\theta_0,x)
  +\bar m_j(\theta_0,x)-\bar m_j(\theta_0,x_0)  \\
&=[\bar m_{\theta,j}(\theta_n^*,x) a]r_n
  +\bar m_j(\theta_0,x)-\bar m_j(\theta_0,x_0)
\le [\bar m_{\theta,j}(\theta_n^*,x) a]r_n
  +C\left(\frac{x-x_0}{\|x-x_0\|}\right)\|x-x_0\|^\gamma
\end{align*}
for $x$ in some neighborhood of $x_0$.  Thus, letting $h$ be some
small scalar going to zero with $n$, for $sh$ and
$(s+t)h$ small enough, we have
\begin{align*}
&E m_j(W_i,\theta_n)I(sh+x_0<X_i<(s+t)h+x_0)  \\
&\le \int_{x\in\mathcal{X}, sh+x_0<x<(s+t)h+x_0}
  \left\{[\bar m_{\theta,j}(\theta_n^*,x) a]r_n
  +C\left(\frac{x-x_0}{\|x-x_0\|}\right)\|x-x_0\|^\gamma\right\}f(x)\,
dx  \\
&=\int_{x_0+uh\in\mathcal{X}, s<u<s+t}
  \left\{[\bar m_{\theta,j}(\theta_n^*,x_0+uh) a]r_n/h^\gamma
  +C\left(\frac{u}{\|u\|}\right)\|u\|^\gamma\right\}f(x_0+uh) h^{d_X+\gamma}\, du
\end{align*}
where the last equality uses the change of variables $u=(x-x_0)/h$.
We also have
\begin{align*}
&\sigma_j^2(sh+x_0,(s+t)h+x_0,\theta_n)  \\
&=E m_j(W_i,\theta_n)^2I(sh+x_0<X_i<(s+t)h+x_0)
-[E m_j(W_i,\theta_n)I(sh+x_0<X_i<(s+t)h+x_0)]^2  \\
&\le \int_{x\in\mathcal{X},sh+x_0<x<(s+t)h+x_0} \mu_{2,j}^2(x)f(x)\, dx
=\int_{x_0+uh\in\mathcal{X} ,s<u<s+t} \mu_{2,j}^2(x_0+uh)f(x_0+uh)h^{d_X}\, dx.
\end{align*}
using the same change of variables.  Under these assumptions,
$\mu_{2,j}^2(x_0+uh)$ converges to $\Sigma_{jj}(x_0)$
and $f(x_0+uh)\to f(x_0)$ uniformly over bounded $u$
as $h$ approaches 0.

Thus, for any $\varepsilon>0$, we will have, for small enough $h$ and
bounded $s$ and $t$,
$E m_j(W_i,\theta_n)I(sh+x_0<X_i<(s+t)h+x_0)/\sigma_j(sh+x_0,th,\theta_n)$
is, for any $s,t$ such that the expression is negative, bounded from
above by
\begin{align*}
&h^{d_X/2+\gamma}\frac{[f(x_0)^{1/2}-\varepsilon]\int_{x_0+uh\in\mathcal{X}, s<u<s+t}
  \left\{[\bar m_{\theta,j}(\theta_0,x_0) a+\varepsilon]r_n/h^\gamma
  +C\left(\frac{u}{\|u\|}\right)\|u\|^\gamma\right\}\, du}
{[\Sigma_{jj}^{1/2}(x_0)+\varepsilon]\text{vol}\{u|x_0+uh\in\mathcal{X},
  s<u<s+t\}^{1/2}}  %
\end{align*}
Setting $h=r_n^{1/\gamma}$, this is equal to
\begin{align*}
r_n^{(d_X/2+\gamma)/\gamma}
\lambda(s,t,(\mathcal{X}-x_0)/r_n^{\gamma},\varepsilon)
\end{align*}
for a function $\lambda$ that does not depend on $r_n$.
Note that the sequence of sets $(\mathcal{X}-x_0)/r_n^{\gamma}$ satisfies
$(\mathcal{X}-x_0)/r_k^{\gamma}\subseteq
(\mathcal{X}-x_0)/r_\ell^{\gamma}$
for $r_\ell<r_k$ by convexity of $\mathcal{X}$,
so, letting $\mathcal{U}=\cup_{k=1}^\infty (\mathcal{X}-x_0)/r_k$, we
will have $\text{vol}(\{s<u<s+t\}\cap (\mathcal{U}\backslash
(\mathcal{X}-x_0)/r_k))\to 0$.  It follows that
$\lambda(s,t,(\mathcal{X}-x_0)/r_n,\varepsilon)\to
\lambda(s,t,\mathcal{U},\varepsilon)$.
Since this holds for all $\varepsilon>0$ and $\lambda$ is continuous in
$\varepsilon$, we have, for any $s,t$,
\begin{align*}
&-\frac{E_n
 m_j(W_i,\theta_n)I(sr_n^{1/\gamma}+x_0<X_i<(s+t)r_n^{1/\gamma}+x_0)}
  {\hat\sigma_j(sr_n^{1/\gamma}+x_0,tr_n^{1/\gamma},\theta_n)}
-\hat q_{1-\alpha}  \\
&=-\frac{Em_j(W_i,\theta_n)I(sr_n^{1/\gamma}+x_0<X_i<(s+t)r_n^{1/\gamma}+x_0)}
  {\sigma_j(sr_n^{1/\gamma}+x_0,tr_n^{1/\gamma},\theta_n)}
+\mathcal{O}_P(n^{-1/2})
-\hat q_{1-\alpha}  \\
&\ge r_n^{(d_X/2+\gamma) /\gamma}[-\lambda(s,t,\mathcal{U},0)+o(1)
-\hat q_{1-\alpha}/r_n^{(d_X/2+\gamma) /\gamma}]
  +\mathcal{O}_P(n^{-1/2})  \\
&=r_n^{(d_X/2+\gamma) /\gamma}\left\{-\lambda(s,t,\mathcal{U},0)+o(1)
-\frac{(2\log t_n^{-d_X})^{1/2}}{n^{1/2}r_n^{(d_X/2+\gamma)/\gamma}}(1+o_P(1))
  +\mathcal{O}_P(n^{-1/2}r_n^{-(d_X/2+\gamma)/\gamma})\right\}
\end{align*}
Note that the $o_P(1)$ term absorbs the $\mathcal{O}_P(1)$ term, so
that the above expression will be negative with probability
approaching one for some $s,t$ as long as
\begin{align*}
\limsup \frac{(2\log t_n^{-d_X})^{1/2}}{n^{1/2}r_n^{(d_X/2+\gamma)/\gamma}}
  < \sup_{s,t} -\lambda(s,t,\mathcal{U},0),
\end{align*}
and this condition can be rearranged to
\begin{align*}
\liminf r_n
\left(\frac{n}{2\log t_n^{-d_X}}\right)^{\gamma/(d_X+2\gamma)}
  > -1/[\inf_{s,t}\lambda(s,t,\mathcal{U},0)]^{\gamma/(d_X/2+\gamma)}.
\end{align*}

\section{Comparison to Intermediate Gaussian Approximations}\label{gauss_approx_sec}

In this section of the appendix, we compare our approach to the results that could be obtained using intermediate gaussian approximations.
As shown in Section \ref{power_sec}, $t_n$ must be chosen at least as small as the optimal bandwidth in order for the test to have good power for a given data generating process.  Theorem \ref{asym_dist_thm_randx} allows $t_n$ to be chosen equal to $n^{-1/d_X}$ times a $\log n$ term, which is small enough to adapt to any Holder class for the conditional mean.
Using the best available results for gaussian approximations in \citet{rio_local_1994} would give a rate of approximation of a $\log n$ term times
$n^{-1/[2(d_X+1]}$
for the random process
$(s,t)\mapsto \sqrt{n} E_n m(W_i,\theta)I(s<X_i<s+t)$.  The test statistic weights this by the inverse of its estimated standard deviation which, at the minimum scale $t_n$, is of order $t_n^{-d_X/2}$.  Thus, in order to use the gaussian approximation of \citet{rio_local_1994}, we would need
$t_n^{-d_X/2}\cdot n^{-1/[2(d_X+1)]}$ to go to zero more quickly than a $\log n$ term, which would mean that
$t_n$ would have to decrease more slowly than a $\log n$ term times
$n^{-\frac{1}{d_X(d_X+1)}}$.
For the test to achieve optimal power when the conditional mean has $\gamma$ conditinuous derivatives (where noninteger values of $\gamma$ corresond to Holder conditions), $t_n$ must decrease at least as quickly as
$n^{-1/(d_X+2\gamma)}$.  Thus, using a gaussian approximation would only lead to optimal power when
$\frac{1}{d_X+2\gamma}
  \le \frac{1}{d_X(d_X+1)}$, which can be rewritten as
$d_X+2\gamma\ge d_X^2+d_X$ or
$\gamma\ge d_X^2/2$.

Since the use of positive kernels (in this case indicator functions) prevents multiscale statistics from being adaptive to $\gamma>2$ derivatives,
this means that
the approach based on the gaussian approximations in \citet{rio_local_1994} %
would be adaptive to a range of $[d_X^2/2,2]$ %
for the smoothness parameter $\gamma$.  Thus, while this approach would lead to useful (if not optimal) results for a one dimensional covariate, %
it would not be adaptive in two dimensions, and would be %
dominated by a kernel statistic with a single bandwidth in more than two dimensions.
In contrast, our result allows adaptivity to all $\gamma$ in $(0,2]$ regardless of the dimension of $X_i$, which is the best possible result.

Another approach is to restrict the set $(s,t)$ over which the supremum is taken to a finite set and place conditions on the rate at which this set increases with the sample size.  While this approach does not apply directly to the statistic considered here, it is useful to compare our results to this approach as well.  Using the results of \citet{chatterjee_simple_2005} along with this approach and a method of proof that avoids deriving an asymptotic distribution, \citet{chetverikov_adaptive_2012} provides a test that is adaptive in the range $\gamma\in[d_X,2]$.
However, using the more recent results of \citet{chernozhukov_gaussian_2013}, which was written contemporaneously with the first draft of the present paper, this could be improved to achieve adaptivity in the full range $\gamma\in (0,2]$.

\section{Power Comparisons with Other Procedures}\label{power_comp_sec}

This appendix discusses in more detail the optimality properties mentioned in the main text.  Since most of the results used here are from other papers, we refer to these papers for details.

\citet{armstrong_choice_2014}, \citet{armstrong_asymptotically_2011} and \citet{armstrong_weighted_2014} show that, under conditions that imply Assumption \ref{local_alt_assump} (these conditions essentially amount to Assumption \ref{local_alt_assump} plus an assumption that $\gamma$ in that condition is the largest $\gamma$ possible), several other procedures do not achieve the same rate for detecting local alternatives.  In particular, the conclusions of Theorem \ref{local_alt_thm} can only hold if the sequence $r_n$ approaches zero at a rate that is slower than the rate given in Theorem \ref{local_alt_thm} by a polynomial factor.

While these conditions are arguably natural in conditional moment inequality models, other procedures will do better in certain cases.  For example, the tests of \citet{andrews_inference_2013} and \citet{lee_testing_2013} will perform better under certain alternatives local to a null under which the contact set has nonzero probability (achieving a $\sqrt{n}$ rate where the test in this paper achieves a $\sqrt{n/\log n}$ rate; see the second display of Theorem B.4 in \citealp{armstrong_weighted_2014} for the latter result).
If one chooses between these conditions using a minimax criterion and smoothness conditions, the test in this paper achieves the optimal rate
among tests available in the literature.

Formally, let $\phi_n(\theta)$ be the test in this paper with $t_n=[(\log n)^5/n]^{1/d_X}$
(i.e. $\phi_n(\theta)=1$ when the test rejects and zero otherwise)
 and asymptotic level $\alpha$
(other choices of $t_n$ would work here as well).
Then, for certain classes of distributions $\mathcal{P}_\gamma$ defined by smoothness conditions and additional regularity conditions,
\begin{align}\label{minimax_rate_eq}
\liminf_n\inf_{P\in\mathcal{P}_\gamma}
\inf_{\theta \text{ s.t. } d(\theta,\theta_0)\ge C^*[(\log n)/n]^{\gamma/(d_X+2\gamma)}
\text{ all } \theta_0\in\Theta_0(P)}
E_P\phi_n(\theta)=1
\end{align}
for some finite constant $C^*$.
For several other tests in the literature, the minimax rate is strictly worse (i.e. the right hand side of the display is zero when $\phi_n(\theta)$ is replaced by one of these tests even if the sequence $[(\log n)/n]^{\gamma/(d_X+2\gamma)}$ is replaced by a sequence that approaches zero at a strictly slower rate).
See Section A.2 of \citet{armstrong_choice_2014} for a formal statement.  (Formally, these results apply to a slightly different test than the one used in this paper, since the truncation is done in a different way.  However, the results can be shown to hold for the test in this paper by following the same arguments.)

Note that the rate $[(\log n)/n]^{\gamma/(d_X+2\gamma)}$ differs from the $n^{-1/2}$ rate obtained for local alternatives in regular point identified models.  This arises from sequences of alternatives of the form considered in Theorem \ref{local_alt_thm}, which are local to a point where the conditional moment inequalities are binding on a zero probability set.  The test developed in this paper is taylored to this case, and it achieves the best rate among tests available in the literature for this form of alternative.  In contrast, for regular point identified models, moment conditions will hold with equality on a positive probability set under the (unique) true value of $\theta$.

To our knowledge, the only other tests in the literature that do not have a strictly worse minimax rate than the tests in this paper in the sense described above are those considered by \citet{armstrong_weighted_2014}, \citet{chetverikov_adaptive_2012} and \citet{chernozhukov_intersection_2013}.  Since the first two papers consider tests that differ from those in the present paper only in implementation and in minor details in the definition of the test, let us compare these tests to those proposed in \citet{chernozhukov_intersection_2013}.  The tests in \citet{chernozhukov_intersection_2013} use the supremum of a kernel based estimate of the conditional mean.  If $\gamma$ is known and used to choose an optimal sequence of bandwidths, this test will achieve (\ref{minimax_rate_eq}).  However, if one uses a sequence of kernels based on an incorrect choice of $\gamma$, the rate will be strictly worse. %
\citep[While these results for kernel estimators have not been shown formally in the literature, Theorem 5.3 in][gives the result for setwise confidence regions and rates of convergence in Hausdorff distance, and the above statements follow from similar arguments]{armstrong_weighted_2014}.
Note, however, that these results hold when positive kernels are used, and that the test of \citet{chernozhukov_intersection_2013} may perform better in situations with more smoothness (larger $\gamma$) when $\gamma$ is known and higher order kernels are used.

Thus, the tests proposed in this paper (along with those in \citealp{armstrong_weighted_2014} and \citealp{chetverikov_adaptive_2012}) are the only tests in the literature that achieve (\ref{minimax_rate_eq}) without knowledge of $\gamma$.  In this sense, these tests are adaptive.  
The results described above consider only the rate (the results show that there exists a $C^*$ such that (\ref{minimax_rate_eq}) holds, but do not give the smallest $C^*$ such that (\ref{minimax_rate_eq}) holds), but it seems likely that the kernel approach of \citet{chernozhukov_intersection_2013} will achieve (\ref{minimax_rate_eq}) with a better constant if prior knowledge of $\gamma$ and other aspects of the data generating process are used to pick the optimal bandwidth.  The comparison of critical values in Section \ref{inference_sec} gives some idea of this.

The results described above consider optimality over a class of tests (which appears to include essentially all tests currently available, at least in the recent econometrics literature on conditional moment inequalities).  One may also ask whether the rate in (\ref{minimax_rate_eq}) is optimal among all tests (i.e., if one replaces $C^*$ with a small enough $C_*>0$, the right hand side of (\ref{minimax_rate_eq}) is 0 for any sequence of level $\alpha$ tests).  While such results are, to our knowledge, not currently available, \citet{menzel_consistent_2010} considers a similar result for the related problem of estimation of the identified set and gives the same rate.

In addition, there is a large literature that considers minimax testing on a conditional mean when $d(\theta,\theta_0)$ is replaced by the distance between $E_P(Y|X_i=x)$ and the $0$ function, where distance is given by the $L_p$ norm (or positive $L_p$ norm) on the space of real valued functions for some $1\le p\le \infty$ \citep[see][for a review of this literature]{ingster_nonparametric_2003}.  These results apply in our setting with $Y_i=m(W_i,\theta)$.  Formally, these papers give constants $0<C_*\le C^*$ and sequences $r_n$ such that
\begin{align*}
\liminf_n\inf_{P\in\mathcal{P}_\gamma \text{ s.t. } \varphi(E_P(Y_i|X_i=x))\ge C_* r_n}
E_P\phi_n=0
\end{align*}
for any sequence of level $\alpha$ tests $\phi_n$ of and, for some sequence of level $\alpha$ tests $\phi_n^*$,
\begin{align}\label{lp_minimax_eq}
\liminf_n\inf_{P\in\mathcal{P}_\gamma \text{ s.t. } \varphi(E_P(Y_i|X_i=x))\ge C^*r_n}
E_P\phi_n=1,
\end{align}
where the null is given by $H_0:E_P(Y_i|X_i)\ge 0$ a.s. or $H_0:E_P(Y_i|X_i)= 0$ a.s.
Here $\varphi$ is a functional from the space of measurable functions on the support of $X_i$ to $[0,\infty)$ that measures distance of the conditional mean from zero, and $r_n$, $C_*$ and $C^*$ depend on $\gamma$ and $\psi$.

A striking finding of this literature is that the rate $r_n$ and optimal test $\phi_n$ depend on the distance $\varphi$.
For the case where
$\varphi(f)=\max\{-\inf_{x\in\mathcal{X}}f(x),0\}$, \citet{dumbgen_multiscale_2001} and \citet{chetverikov_adaptive_2012} give these results with $r_n$ given by the rate in (\ref{minimax_rate_eq}), and show that tests similar to those considered in the present paper achieve this rate.  For the two-sided version of this problem, \citet{spokoiny_adaptive_1996} and \citet{horowitz_adaptive_2001} have considered adaptivity under the $L_p$ norm $\varphi(f)=(\int |f(x)|^p\,d\mu(x))^{1/p}$ with $p<\infty$.  In contrast to the $L_\infty$ case, \citet{horowitz_adaptive_2001} use a statistic that takes the supremum over bandwidths of a test based on the $L_2$ norm of a kernel estimate of the conditional mean for the case where $p=2$.  By the results in \citet{armstrong_choice_2014} mentioned above, the generalization of this statistic to the one-sided case with $\gamma$ known \citep[which corresponds to one of the statistics considered by][]{lee_testing_2013} has a worse rate when one considers distance on $\theta$ with $Y_i=m(W_i,\theta)$ as in (\ref{minimax_rate_eq}).

Thus, the results in \citet{dumbgen_multiscale_2001}, \citet{chetverikov_adaptive_2012}, \citet{armstrong_weighted_2014} and \citet{armstrong_choice_2014} discussed above suggest a connection between Euclidean distance on $\theta$ and the $L_\infty$ norm for the conditional mean in these two problems.  Our test is geared toward achieving good rates in (\ref{minimax_rate_eq}).  This reflects the practice of inverting hypothesis tests to obtain a confidence region for points in $\Theta_0(P)$ \citep[see][]{imbens_confidence_2004}.  The rates in (\ref{minimax_rate_eq}) reflect how fast this confidence region shrinks toward $\Theta_0(P)$.  Indeed, there is a close connection between this notion of relative efficiency and Hausdorff distance on sets in $\Theta$, and \citet{armstrong_weighted_2014} shows that a version of our test with a stronger notion of coverage achieves the same rate of convergence in the Hausdorff metric between the confidence region and identified set.
If one is not interested in $\theta$ and cares instead about detecting conditional means that violate the null by a particular amount according to an $L_p$ norm, the optimal test will depend on $p$ and will be different in the case where $p<\infty$.

\section{Other Methods of Calculating Critical Values}\label{other_cval_sec}

This appendix discusses other methods for computing critical values for our test.  Results in the literature for other settings suggest that these methods may provide an improvement to the higher order accuracy of the nominal coverage of the test, particularly in the case where $t_n^{d_X}$ is not too small relative to $\text{vol}(\mathcal{X})$ \citep[see][and the discussion below]{hall_rate_1979,piterbarg_asymptotic_1996}.  However, we leave the question of higher order coverage accuracy for future research.

\subsection{Direct Application of Tail Approximations}\label{direct_approx_sec}

Our asymptotic distribution result uses a tail approximation of the form
\begin{align}\label{tail_approx_refinement_eq}
P\left(\sqrt{n}S_{n}\le r_n\right)
\sim \exp\left(-d_Y \text{vol}(\mathcal{X})
t_n^{-d_X} \exp\left(-r_n^2/2\right)r_n^{4d_X-1}\pi^{-1/2}2^{-2d_X-1/2}
\right)
\end{align}
for the ``least favorable'' case where $E(m(W_i,\theta)|X_i)=0$ almost surely.  The result follows by setting $r_n=\sqrt{n}(r+b(\hat c_n))/a(\hat c_n)$, and noting that, for this choice of $r_n$ the right hand side of the above display converges in probability to the extreme value cdf $\exp(-d_X\exp(-r))$.

This suggests another approach to calculating critical values: choose the critical value based directly on the right hand side of (\ref{tail_approx_refinement_eq}).  That is, we reject when $\sqrt{n}S_n>\hat q_\alpha$ where $\hat q_\alpha$ is the largest solution to
\begin{align*}
1-\alpha=\exp\left(-d_Y \text{vol}(\hat{\mathcal{X}})
t_n^{-d_X} \exp\left(-q^2/2\right)q^{4d_X-1}\pi^{-1/2}2^{-2d_X-1/2}\right).
\end{align*}
\citet{piterbarg_asymptotic_1996} suggested a version of this approach in settings with a stationary Gaussian process, and showed that it leads to an asymptotic refinement in that setting in the sense that the critical value corresponding to our critical value in the main text gives a test with size $\alpha+\mathcal{O}(1/\log n)$, while the approach described in this section gives a test with size $\alpha+\mathcal{O}(n^{-K})$ for some constant $K$.
In a different setting involving the supremum of a process exhibiting a different type of nonstationarity, \citet{lee_testing_2009} propose a similar correction.  While those authors do not formally consider asymptotic refinements, they provide monte carlo evidence of an improvement in finite samples.
Based on these results, it seems likely that this correction could be shown to lead to some improvement in our setting.  Formalizing these ideas is a useful direction for future research.

\subsection{Simulated Critical Values}\label{sim_cval_sec}

Our results can also be used to show the asymptotic validity of simulated critical values based on a certain bootstrap procedure.  Since this procedure is based directly on the supremum of a random process rather than an extreme value limit, one might expect this procedure to give an improvement in coverage accuracy.  We leave this question for future research, although we note that \citet{chetverikov_adaptive_2012} has shown that a different bootstrap procedure applied to a closely related test statistic achieves polynomial coverage.  Whether this applies in our setting, and whether the polynomial rate is better than the one achieved by the refinement in Section \ref{direct_approx_sec} (if this refinement does indeed achieve a polynomial rate), are both interesting questions for future research.

We define our simulated critical values as follows.  For each $j$, let $\hat M_{n}(x)$ be any random sequence of functions that take values in $\mathcal{X}$ to $d_Y\times d_Y$ symmetric, positive definite, matrices.  We require that sequence of variance matrices given by $\hat M_n(x)$ be continuous in $x$ and have correlation coefficients bounded away from one uniformly over $n$ with probability one.
One can choose $\hat M_{n}(x)$ to be an estimate of the conditional variance matrix of the $m(W_i,\theta)$, but this is not necessary, and $\hat M_{n}(x)$ can even be chosen to be the constant function that takes all values to the identity matrix.  For each repetition $b$ of $B$ simulations, we draw $n$ independent outcome variables
$\{Y_{i}^{*,b}\}_{i=1}^n$ with $Y_{i}^{*,b}\sim N(0,\hat M_n(X_i))$ independent across $i$ and $b$ conditional on the data.  We form the test statistic $S_{n,b}^*$ for this repetition by replacing $m(W_i,\theta)$ with $Y_{i}^{*,b}$ in the definition of the test statistic.  The simulated critical value is given by the $1-\alpha$ quantile of this bootstrap distribution:
\begin{align}\label{cval_sim_eq}
\hat q_{1-\alpha,\text{sim}} = \inf \left\{r\Bigg|\frac{1}{B} \sum_{b=1}^B I(S_{n,b}^*\le r)\ge 1-\alpha\right\}.
\end{align}
The asymptotic validity of this test follows immediately from the version of Theorem \ref{asym_dist_thm_randx} in Appendix \ref{uniformity_sec} that incorporates uniformity in the underlying distribution.

\begin{theorem}\label{test_size_sim_thm}
Suppose that the null hypothesis (\ref{null_eq}) holds for $\theta$ and that Assumption \ref{asym_dist_assump} holds.  Let
$\hat q_{1-\alpha,\text{sim}}$ be as defined in (\ref{cval_sim_eq})
where
$\hat M_n(x)$ is continuous in $x$ uniformly in $n$ and has correlation coefficients bounded away from one uniformly over $n$ with probability one.
Then
\begin{align*}
\limsup_n P\left(S_n(\theta)>\hat q_{1-\alpha,\text{sim}}\right)
  \le \alpha.
\end{align*}
If, in addition, $\bar m(\theta,x)=0$ for all $x\in\mathcal{X}$, then
\begin{align*}
P\left(S_n(\theta)>\hat q_{1-\alpha,\text{sim}}\right)
  \to \alpha.
\end{align*}
\end{theorem}

\section{Additional Monte Carlos}\label{additional_mc_sec}

This section describes the details for the 
monte carlo analysis of the test of \citet{chetverikov_adaptive_2012}
reported in Tables \ref{power_table_ch_d1}, \ref{power_table_ch_d2} and \ref{power_table_ch_d3} and described briefly in Section \ref{monte_carlo_sec}.
We include two forms of the test that differ in the choice of critical value.  The plug-in asymptotic (PIA) critical value uses the least favorable distribution to compute the critical value.  The generalized moment selection (GMS) critical value uses a pre-test to determine which moments are close to binding.  We describe the additional implementation details below.

The test involves choosing
the kernel $K$, a set of bandwidths $H$ and an estimate $\hat \Sigma_{j,11}$ of the conditional variance $var(m(W_i,\theta)|X_i=x)$ evaluated at $x=X_j$.  We follow
\citeauthor{chetverikov_adaptive_2012}'s recommendations %
of setting $K(x)=.75(1-\|x\|^2)I(\|x\|\le 1)$ and $H=\{h=h_{\text{max}}a^k|g\ge h_{\text{min}}, k=0,1,2,3,\ldots\}$ with
$a=.5$, $h_{\text{max}}=\max_{i,j=1,\ldots,n}\|X_i-X_j\|/2$,
$h_{\text{min}}=0.2h_{\text{max}}(\log n/n)^{1/3}$.
For the variance estimate, we form
$\hat\Sigma_{j,11}$
as a nearest neighbor estimator, given by the sample variance of the $Y_i's$ corresponding to the observations with the nearest $k_n= \lceil n^{2/3}\rceil$
values of $X$ to $X_j$ (\citeauthor{chetverikov_adaptive_2012} discusses variance estimation %
but does not explicitly recommend a particular estimator).
For the critical value, we consider the plug-in asymptotic critical value $c_{1-\alpha}^{PIA}$ as well as the refined moment selection critical value $c_{1-\alpha}^{RMS}$, formed as in %
\citet{chetverikov_adaptive_2012} using $\gamma=0.1/\log(n)$ as discussed by \citeauthor{chetverikov_adaptive_2012}.
We use $500$ monte carlo replications and $500$ bootstrap replications to form the critical value for each monte carlo replication.

To describe the test statistic, let us relabel the indices of the observations so that $X_1<X_2<\ldots<X_n$.
The test statistic is given by
\begin{align*}
T=\max_{s\in S} \frac{\hat f_s}{\hat V_s}
\end{align*}
where $S=\{(i,h)|i=1,n/50+1,2n/50+1,\ldots,(n-1)/50+1, h\in H\}$,
$\hat f_{i,h}=\frac{\sum_{j=1}^n K((X_i-X_j)/h)Y_j}{\sum_{j=1}^n K((X_i-X_j)/h)}$
and
$\hat V_{i,h}^2=\frac{\sum_{j=1}^n K((X_i-X_j)/h)^2\hat \Sigma_{j,11}^2}{\left[\sum_{j=1}^n K((X_i-X_j)/h)\right]^2}$
where $\hat \Sigma_{j,11}$ is the variance estimator described above.
(This differs from the choice of $S$ used in \citet{chetverikov_adaptive_2012} in that the maximum of the kernel estimator is taken at 50 values of $x$ corresponding to a subset of $X_i$ values rather than all $n$ values of $X_i$.  This is done for reasons of computational feasibility.)
The critical value $c_{1-\alpha}^{PIA}$ is given by the $1-\alpha$ quantile of 
\begin{align*}
T^{PIA}=
\max_{(i,h)\in S}\hat V_{i,h}^{-1}\frac{\sum_{j=1}^n K((X_i-X_j)/h)\tilde Y_j}{\sum_{j=1}^n K((X_i-X_j)/h)}
\end{align*}
with $\tilde Y_j\sim N(0,\hat\Sigma_{1,jj})$ independently across $j$.  The critical value $c_{1-\alpha}^{RMS}$ is given by the $1-\alpha$ quantile of
\begin{align*}
T^{RMS}=
\max_{(i,h)\in S^{RMS}}\hat V_{i,h}^{-1}\frac{\sum_{j=1}^n K((X_i-X_j)/h)\tilde Y_j}{\sum_{j=1}^n K((X_i-X_j)/h)}
\end{align*}
where
$S^{RMS}=\{s|\hat f_s/\hat V_s>-2 c_{1-\gamma}^{PIA}\}$.
We report the results of this test for the original monte carlo designs with both choices of critical value (PIA and RMS).

\bibliography{library}

\newpage

\begin{figure}[h]
  \centering
  \caption{95\% Confidence Region Using Weighted Sup Statistic (this paper)}
  \includegraphics[height=2.8in]{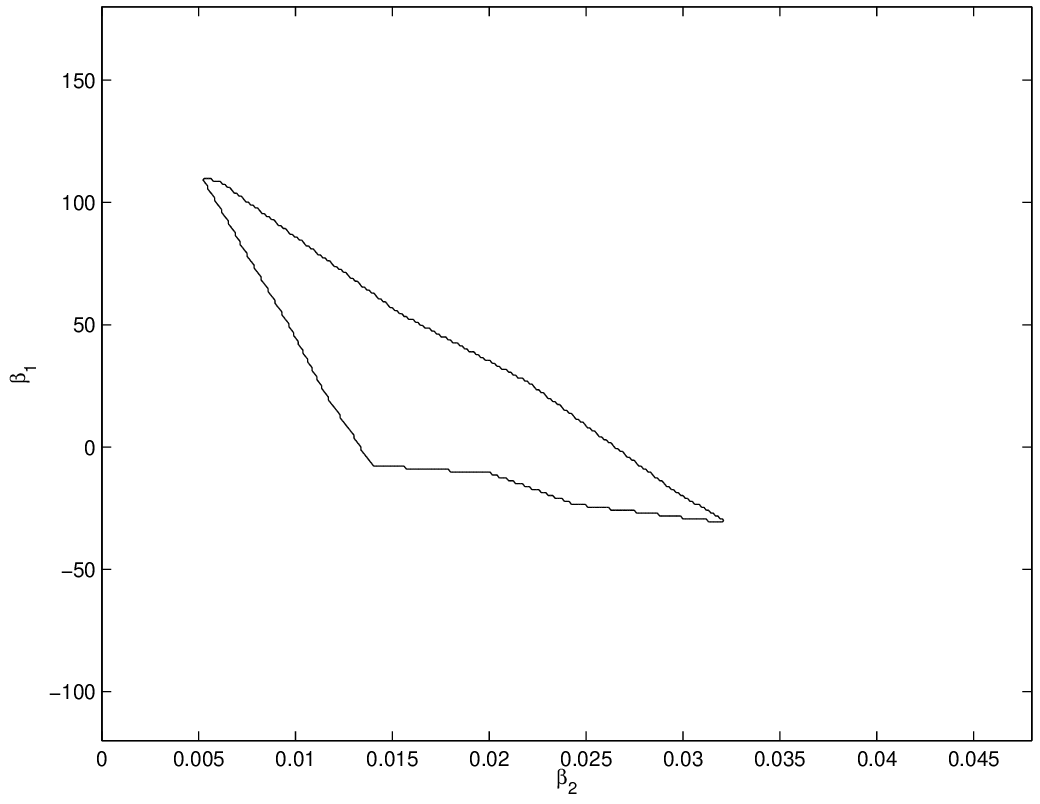}
    \label{cr95_asym_dist_fig}
\end{figure}

\begin{figure}[h]
  \centering
  \caption{95\% Confidence Region Using Unweighted Statistic and Subsampling with Estimated Rate}
  \includegraphics[height=2.8in]{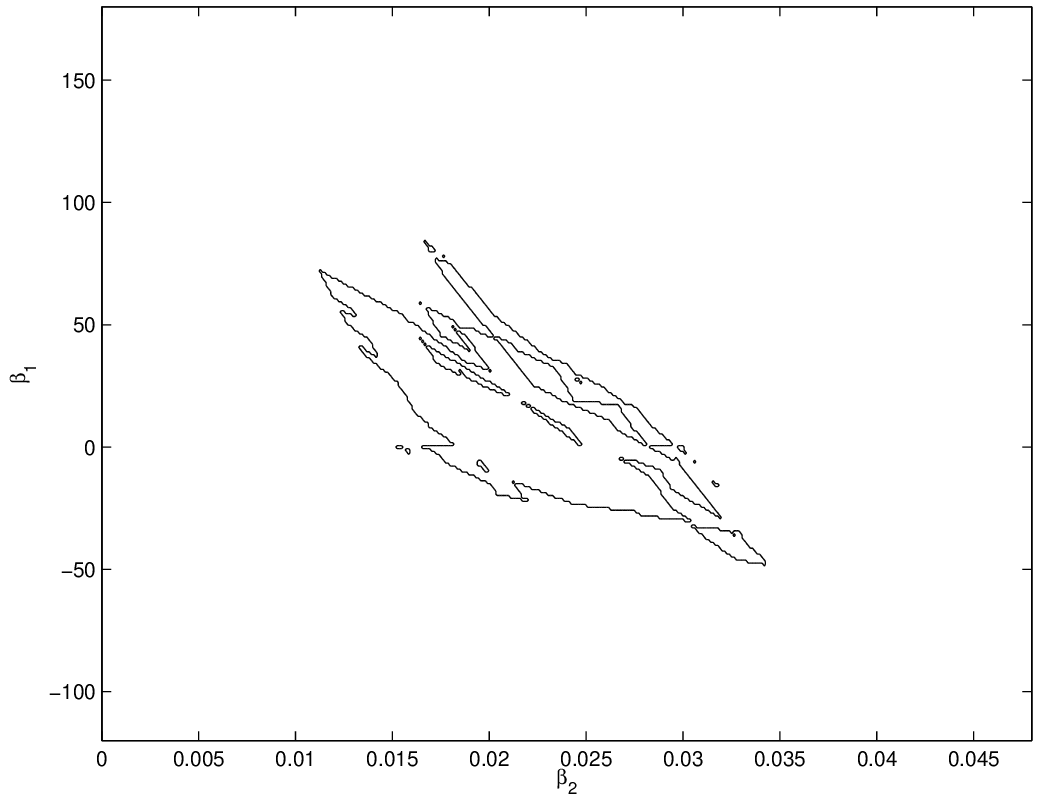}
    \label{cr95_est_fig}
\end{figure}

\begin{figure}[h]
  \centering
  \caption{95\% Confidence Region Using Unweighted Statistic and Subsampling with Conservative Rate}
  \includegraphics[height=2.8in]{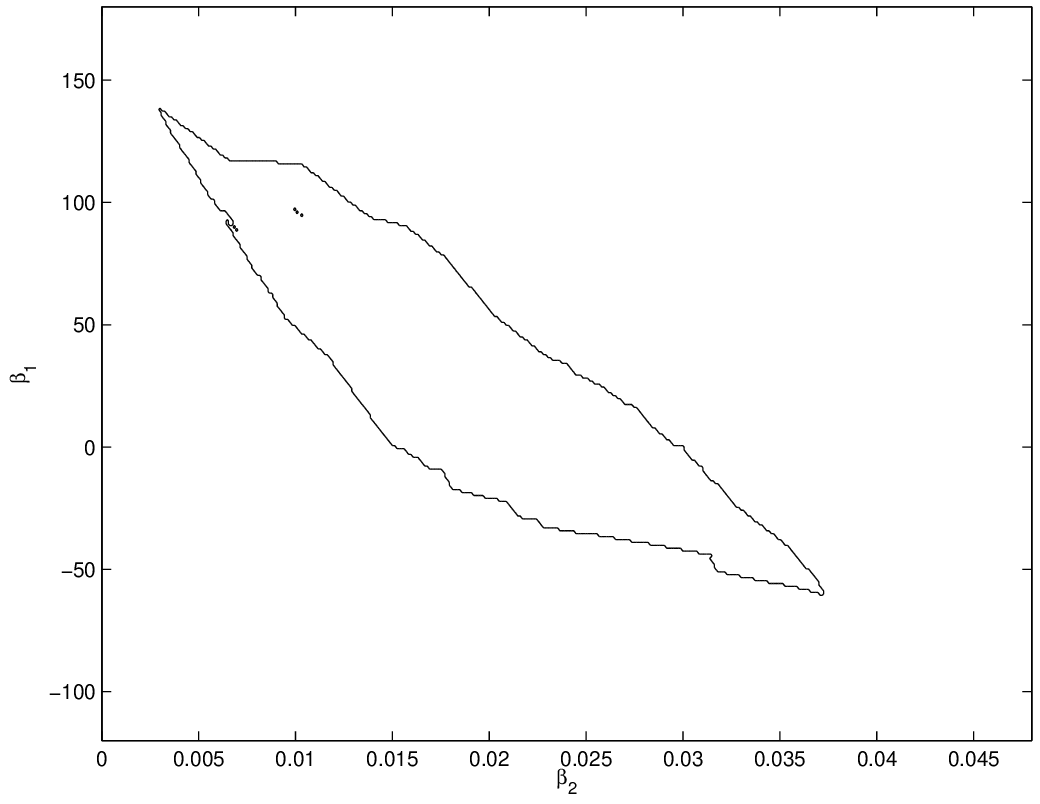}
    \label{cr95_cons_fig}
\end{figure}

\clearpage
\newpage

\begin{table}
\begin{center}
\begin{tabular}{c|c|ccc}
& $t_n$ & $n=100$ & $n=500$ & $n=1000$  \\
\hline
\multirow{3}{*}{nominal size $.1$}& $n^{-1/5}$ &     0.2510  &  0.1840  &  0.1770 \\
& $n^{-1/3}$ & 0.1640  &  0.1160  &  0.1150  \\
& $n^{-1/2}$ &     0.0890  &  0.0770  &  0.0880  \\
\hline 
\multirow{3}{*}{nominal size $.05$} & $n^{-1/5}$ &  0.1020  &  0.0650  &  0.0790 \\
& $n^{-1/3}$ &     0.0750  &  0.0410  &  0.0550  \\
& $n^{-1/2}$ &     0.0340  &  0.0220  &  0.0350
\end{tabular}
\caption{False Rejection Probabilities for Least Favorable Null}
\label{false_rejection_prob_table}
\end{center}
\end{table}

\begin{table}
\begin{center}
\begin{tabular}{c|c|ccc}
 $t_n$ & $\theta_1-\overline\theta_1$ & $n=100$ & $n=500$ & $n=1000$  \\
\hline
\multirow{6}{*}{$n^{-1/5}$}
& 0  &    0.0490  &  0.0490  &  0.0500   \\
& .1 &    0.2070  &  0.5030  &  0.7290   \\
& .2 &    0.4800  &  0.9540  &  1.0000   \\ 
& .3 &    0.7590  &  1.0000  &  1.0000   \\
& .4 &    0.9560  &  1.0000  &  1.0000   \\
& .5 &    0.9970  &  1.0000  &  1.0000   \\
\hline
\multirow{6}{*}{$n^{-1/3}$}
& 0  &    0.0500  &  0.0500  &  0.0500   \\
& .1 &    0.1440  &  0.4530  &  0.6300  \\
& .2 &    0.3780   & 0.9390  &  0.9980  \\
& .3 &    0.6910  &  1.0000  &  1.0000  \\
& .4 &    0.8860  &  1.0000  &  1.0000  \\
& .5 &    0.9820  &  1.0000 &   1.0000  \\
\hline 
\multirow{6}{*}{$n^{-1/2}$}
& 0  &    0.0440  &  0.0500  &  0.0490   \\
& .1 &      0.1560  &  0.3580  &  0.5020   \\
& .2 &     0.3480  &  0.8980  &  0.9910   \\
& .3 &     0.6490  &  0.9990  &  1.0000   \\
& .4 &     0.8620  &  1.0000  &  1.0000   \\
& .5 &     0.9740  &  1.0000  &  1.0000
\end{tabular}
\caption{Power for Level $\alpha=.05$ Test with
  Critical Values Based on Finite Sample Least Favorable Distribution (Design 1)}
\label{power_table_d1}
\end{center}
\end{table}

\begin{table}
\begin{center}
\begin{tabular}{c|c|ccc}
 $t_n$ & $\theta_1-\overline\theta_1$ & $n=100$ & $n=500$ & $n=1000$  \\
\hline
\multirow{6}{*}{$n^{-1/5}$}
& 0  &    0.0020  &  0.0000  &  0.0000   \\
& .1 &   0.0000  &   0.0000  &   0.0000    \\
&.2 &    0.0060 &   0.0160  &  0.0320  \\
&.3&    0.0260  &  0.1380  &  0.2950  \\
&.4&    0.0640  &  0.4490  &  0.8310  \\
&.5&    0.1750  &  0.8480  &  0.9950  \\
\hline
\multirow{6}{*}{$n^{-1/3}$}
& 0  &    0.0020  &  0.0000  &  0.0010   \\
& .1 &    0.0070  &  0.0120 &   0.0050  \\
& .2 &    0.0160  &  0.0620 &   0.1000  \\
& .3 &    0.0410  &  0.2150 &   0.4560  \\
& .4 &    0.1190  &  0.6040 &   0.8760  \\
& .5 &    0.2100  &  0.9020 &   0.9960  \\
\hline 
\multirow{5}{*}{$n^{-1/2}$}
& 0  &    0.0020  &  0.0010  &  0.0010   \\
& .1 &    0.0060  &  0.0140 &   0.0100   \\
& .2 &    0.0230  &  0.0570 &   0.0860   \\
& .3 &    0.0380  &  0.2290 &   0.3890   \\
& .4 &    0.1190  &  0.5320 &   0.7910   \\
& .5 &    0.2030  &  0.8500 &   0.9820
\end{tabular}
\caption{Power for Level $\alpha=.05$ Test with
  Critical Values Based on Finite Sample Least Favorable Distribution (Design 2)}
\label{power_table_d2}
\end{center}
\end{table}

\begin{table}
\begin{center}
\begin{tabular}{c|c|ccc}
 $t_n$ & $\theta_1-\overline\theta_1$ & $n=100$ & $n=500$ & $n=1000$  \\
\hline
\multirow{6}{*}{$n^{-1/5}$}
& 0  &    0.0100  &  0.0030  &  0.0020   \\
& .1 &    0.0340 &   0.0640 &   0.1200  \\
& .2 &    0.0930 &   0.4660 &   0.7040  \\
& .3 &    0.2720 &   0.8690 &   0.9900  \\
& .4 &    0.5010 &   0.9940 &   1.0000  \\
& .5 &    0.7670 &   1.0000 &   1.0000  \\
\hline
\multirow{6}{*}{$n^{-1/3}$}
& 0  &    0.0140  &  0.0040  &  0.0070   \\
& .1 &    0.0390  &  0.1040  &  0.1160  \\
& .2 &    0.1120  &  0.4290  &  0.6400  \\
& .3 &    0.2570  &  0.8380  &  0.9790  \\
& .4 &    0.4630  &  0.9940  &  1.0000  \\
& .5 &    0.7170  &  1.0000  &  1.0000  \\
\hline 
\multirow{5}{*}{$n^{-1/2}$}
& 0  &    0.0160  &  0.0060  &  0.0080   \\
& .1 &    0.0300  &  0.0830 &   0.0870  \\
& .2 &    0.1210  &  0.3250 &   0.5230  \\
& .3 &    0.2400  &  0.7620 &   0.9670  \\
& .4 &    0.3970  &  0.9840 &   1.0000  \\
& .5 &    0.6690  &  1.0000 &   1.0000
\end{tabular}
\caption{Power for Level $\alpha=.05$ Test with
  Critical Values Based on Finite Sample Least Favorable Distribution (Design 3)}
\label{power_table_d3}
\end{center}
\end{table}

\begin{table}
\begin{center}
\begin{tabular}{c|c|ccc}
 $t_n$ & $\theta_1-\overline\theta_1$ & $n=100$ & $n=500$ & $n=1000$  \\
\hline
\multirow{5}{*}{$n^{-1/5}$}
&  0 &     0.0240  &  0.0250  &  0.0250      \\
& .1 &     0.1260  &  0.4200  &  0.5980      \\
& .2 &     0.3460  &  0.9350  &  0.9980      \\ 
& .3 &     0.6410  &  1.0000  &  1.0000      \\
& .4 &     0.9150  &  1.0000  &  1.0000      \\
& .5 &     0.9900  &  1.0000  &  1.0000      \\
\hline
\multirow{5}{*}{$n^{-1/3}$}
&  0 &   0.0250  &  0.0250  &  0.0250     \\
& .1 &   0.0780  &  0.3360  &  0.5090     \\
& .2 &   0.2580  &  0.8860  &  0.9930     \\
& .3 &   0.5690  &  0.9990  &  1.0000     \\
& .4 &   0.8140  &  1.0000  &  1.0000     \\
& .5 &   0.9650  &  1.0000  &  1.0000     \\
\hline 
\multirow{5}{*}{$n^{-1/2}$}
&  0 &    0.0230 &   0.0250  &  0.0250       \\
& .1 &    0.0960 &   0.2630  &  0.3770       \\
& .2 &    0.2580 &   0.8430  &  0.9840       \\
& .3 &    0.5440 &   0.9980  &  1.0000       \\
& .4 &    0.7970 &   1.0000  &  1.0000       \\
& .5 &    0.9580 &   1.0000  &  1.0000    
\end{tabular}
\caption{Power for Level $\alpha=1-(1-.05)^{1/2}$ Test (corresponds to level $\alpha$ with a nonbinding moment) with
  Critical Values Based on Finite Sample Least Favorable Distribution (Design 1)}
\label{power_table_em_d1}
\end{center}
\end{table}

\begin{table}
\begin{center}
\begin{tabular}{c|c|ccc}
 $t_n$ & $\theta_1-\overline\theta_1$ & $n=100$ & $n=500$ & $n=1000$  \\
\hline
\multirow{5}{*}{$n^{-1/5}$}
&  0 &      0     &     0       &    0         \\
& .1 &      0     &     0       &    0         \\
& .2 &     0.0030 &     0.0090  &    0.0110   \\ 
& .3 &     0.0130 &     0.1040  &    0.2020   \\
& .4 &     0.0270 &     0.3860  &    0.7450   \\
& .5 &     0.1120 &     0.7880  &    0.9900   \\
\hline
\multirow{5}{*}{$n^{-1/3}$}
&  0 &     0.0020  &       0  &       0  \\
& .1 &     0.0020  &  0.0060  &  0.0040  \\
& .2 &     0.0070  &  0.0330  &  0.0670  \\
& .3 &     0.0190  &  0.1370  &  0.3710  \\
& .4 &     0.0710  &  0.4920  &  0.8270  \\
& .5 &     0.1370  &  0.8270  &  0.9930  \\
\hline 
\multirow{5}{*}{$n^{-1/2}$}
&  0 &    0.0010 &        0 &   0.0010     \\
& .1 &    0.0010 &   0.0090 &   0.0030      \\
& .2 &    0.0100 &   0.0400 &   0.0400      \\
& .3 &    0.0250 &   0.1640 &   0.2990      \\
& .4 &    0.0760 &   0.4540 &   0.7210      \\
& .5 &    0.1440 &   0.7930 &   0.9730   
\end{tabular}
\caption{Power for Level $\alpha=1-(1-.05)^{1/2}$ Test (corresponds to level $\alpha$ with a nonbinding moment) with
  Critical Values Based on Finite Sample Least Favorable Distribution (Design 2)}
\label{power_table_em_d2}
\end{center}
\end{table}

\begin{table}
\begin{center}
\begin{tabular}{c|c|ccc}
 $t_n$ & $\theta_1-\overline\theta_1$ & $n=100$ & $n=500$ & $n=1000$  \\
\hline
\multirow{5}{*}{$n^{-1/5}$}
&  0 &    0.0060  &  0.0020 &   0.0010   \\
& .1 &    0.0170  &  0.0460 &   0.0680   \\
& .2 &    0.0540  &  0.3750 &   0.5810   \\ 
& .3 &    0.1860  &  0.8260 &   0.9770   \\
& .4 &    0.3810  &  0.9890 &   1.0000   \\
& .5 &    0.6440  &  1.0000 &   1.0000   \\
\hline
\multirow{5}{*}{$n^{-1/3}$}
&  0 &      0.0080 &   0.0010  &  0.0030  \\
& .1 &      0.0210 &   0.0650  &  0.0760  \\
& .2 &      0.0470 &   0.3150  &  0.5620  \\
& .3 &      0.1650 &   0.7520  &  0.9630  \\
& .4 &      0.3430 &   0.9860  &  1.0000  \\
& .5 &      0.5990 &   1.0000  &  1.0000  \\
\hline 
\multirow{5}{*}{$n^{-1/2}$}
&  0 &     0.0090 &   0.0010 &   0.0030    \\
& .1 &     0.0170 &   0.0510 &   0.0510    \\
& .2 &     0.0870 &   0.2480 &   0.4160    \\
& .3 &     0.1590 &   0.6910 &   0.9470    \\
& .4 &     0.2990 &   0.9680 &   1.0000    \\
& .5 &     0.5640 &   0.9990 &   1.0000 
\end{tabular}
\caption{Power for Level $\alpha=1-(1-.05)^{1/2}$ Test (corresponds to level $\alpha$ with a nonbinding moment) with
  Critical Values Based on Finite Sample Least Favorable Distribution (Design 3)}
\label{power_table_em_d3}
\end{center}
\end{table}

\begin{table}
\begin{center}
\begin{tabular}{c|c|ccc}
 Critical Value & $\theta_1-\overline\theta_1$ & $n=100$ & $n=500$ & $n=1000$  \\
\hline
\multirow{5}{*}{PIA}
&  0 &     0.0540  &  0.0680 &   0.0520    \\
& .1 &     0.1680  &  0.3360 &   0.5120    \\
& .2 &     0.3880  &  0.8780 &   0.9980    \\ 
& .3 &     0.6580  &  0.9980 &   1.0000    \\
& .4 &     0.8760  &  1.0000 &   1.0000    \\
& .5 &     0.9760  &  1.0000 &   1.0000    \\
\hline
\multirow{5}{*}{GMS}
&  0 &     0.0540  &  0.0680  &  0.0520    \\
& .1 &     0.1680  &  0.3360  &  0.5120    \\
& .2 &     0.3880  &  0.8780  &  0.9980    \\
& .3 &     0.6580  &  0.9980  &  1.0000    \\
& .4 &     0.8760  &  1.0000  &  1.0000    \\
& .5 &     0.9760  &  1.0000  &  1.0000 
\end{tabular}
\caption{Power for \citet{chetverikov_adaptive_2012} Test (Design 1)}
\label{power_table_ch_d1}
\end{center}
\end{table}

\begin{table}
\begin{center}
\begin{tabular}{c|c|ccc}
 Critical Value & $\theta_1-\overline\theta_1$ & $n=100$ & $n=500$ & $n=1000$  \\
\hline
\multirow{5}{*}{PIA}
&  0 &     0.0020  &  0.0060 &   0.0020   \\
& .1 &     0.0060  &  0.0120 &   0.0220   \\
& .2 &     0.0300  &  0.0680 &   0.0900   \\ 
& .3 &     0.0820  &  0.1880 &   0.4280   \\
& .4 &     0.1420  &  0.5220 &   0.8360   \\
& .5 &     0.2560  &  0.8600 &   0.9940   \\
\hline
\multirow{5}{*}{GMS}
&  0 &     0.0020  &   0.0080  &   0.0040     \\
& .1 &     0.0060  &   0.0140  &   0.0300     \\
& .2 &     0.0300  &   0.0740  &   0.1180     \\
& .3 &     0.0840  &   0.2060  &   0.4680     \\
& .4 &     0.1480  &   0.5400  &   0.8620     \\
& .5 &     0.2580  &   0.8800  &   0.9960  
\end{tabular}
\caption{Power for \citet{chetverikov_adaptive_2012} Test (Design 2)}
\label{power_table_ch_d2}
\end{center}
\end{table}

\begin{table}
\begin{center}
\begin{tabular}{c|c|ccc}
 Critical Value & $\theta_1-\overline\theta_1$ & $n=100$ & $n=500$ & $n=1000$  \\
\hline
\multirow{5}{*}{PIA}
&  0 &     0.0120   &  0.0140  &   0.0100   \\
& .1 &     0.0260   &  0.0500  &   0.1160   \\
& .2 &     0.1080   &  0.3120  &   0.5740   \\ 
& .3 &     0.2660   &  0.7520  &   0.9720   \\
& .4 &     0.4520   &  0.9720  &   1.0000   \\
& .5 &     0.6980   &  0.9980  &   1.0000   \\
\hline
\multirow{5}{*}{GMS}
&  0 &     0.0120  &  0.0180  &  0.0120    \\
& .1 &     0.0260  &  0.0560  &  0.1240    \\
& .2 &     0.1080  &  0.3240  &  0.5900    \\
& .3 &     0.2660  &  0.7560  &  0.9740    \\
& .4 &     0.4520  &  0.9720  &  1.0000    \\
& .5 &     0.6980  &  0.9980  &  1.0000 
\end{tabular}
\caption{Power for \citet{chetverikov_adaptive_2012} Test (Design 3)}
\label{power_table_ch_d3}
\end{center}
\end{table}

\clearpage
\newpage

\begin{table}
\centering
\begin{tabular}{c|c|c}
 & $\theta_1$ & $\theta_2$ \\\hline
Weighted Sup Statistic (this paper) &  $[-30, 109]$
  & $[0.0053, 0.0320]$  \\
Unweighted, Subsampling with Estimated Rate & $[-48, 84]$ & $[0.0113, 0.0342]$  \\
Unweighted, Subsampling with Conservative Rate & $[-60, 138]$ & $[0.0030, 0.0372]$  \\
\hline
\end{tabular}
\caption{95\% Confidence Intervals for Components of $\theta$}
\label{ci_table}
\end{table}

\end{document}